\documentclass[12pt]{article}
\usepackage{amsmath,amssymb,amsthm}
\usepackage[square,comma,numbers,sort&compress]{natbib}
\usepackage[all]{xy}
\newtheorem{thm}{Theorem}
\newtheorem{defi}{Definition}
\newtheorem{coro}{Corollary}
\newtheorem{prop}{Proposition}
\newtheorem{lemma}{Lemma}

\theoremstyle{remark}
\newtheorem{ex}{Example}
\newtheorem{rem}{Remark}
\def\povm{\mathsf M}
\def\Le{\mathcal L}
\def\Ae{\mathcal A}

\def\Oe{\mathcal O}
\def\Ve{\mathcal V}
\def\Ce{\mathcal C}
\def\Ee{\mathcal E}\def\Fe{\mathcal F}

\def\Te{\mathcal T}

\def\states{\mathfrak S}
\def\Se{\mathcal S}
\def\Pe{\mathcal P}
\def\Be{\mathcal B}
\def\De{\mathcal D}
\def\Tr{\mathrm{Tr}\,}
\def\ptr{\mathrm{Tr}}
\def\<{\langle}
\def\>{\rangle}
\def\diag{\mathbb{D}}
\title{Randomization theorems for quantum channels}
\author{Anna Jen\v cov\'a\footnote{jenca@mat.savba.sk}\\ Mathematical Institute, Slovak Academy of Sciences\\
\v Stef\'anikova 49, Bratislava, Slovakia}
\date{}
\begin{document}
\maketitle
\abstract{The classical randomization criterion is an important result of statistical decision theory. Recently, a quantum analogue has been proposed, giving equivalent conditions for two sets of quantum states,  ensuring existence of a quantum channel mapping one set close to the other, in $L_1$-distance.   In the present paper, we extend these concepts in several ways. First, sets of states are replaced by channels and randomization is performed by either post- or pre-composition by another channel.  The $L_1$-distance is replaced by the diamond norm. Secondly, the maps are not required to be completely positive, but positivity is given by an admissible  family of convex cones. It is shown that the randomization theorems, generalizing both quantum and classical randomization criteria, can be proved in the framework of base section norms, including the diamond norm and its dual. The theory of such norms is developed in the Appendix.  }
\section{Introduction}

The classical randomization criterion \cite{strasser1985statistics,torgersen1991comparison} for statistical experiments is an important result of statistical decision theory. 
 It makes a link between \emph{comparison of experiments} in terms of payoffs of decision rules for 
decision problems and the $L_1$-distance of their  \emph{randomizations}. In particular, the Blackwell-Sherman-Stein (BSS) \cite{blackwell1951comparison, sherman1951theorem, stein1951comparison} theorem gives 
equivalent conditions for one experiment to be a randomization of the other, that is, for existence of a stochastic mapping (Markov kernel) transforming one experiment into the other.

Quantum versions of these results were recently studied by several authors, \cite{buscemi2012comparison, matsumoto2010randomization,shmaya2005comparison,jencova2012comparison}. 
A quantum statistical experiment is a parametrized (usually finite) family  of density operators on some (usually finite dimensional) Hilbert space. Stochastic mappings are  replaced  by completely positive trace preserving maps, 
or \emph{channels}, but other choices are possible, e.g. positive trace preserving maps.  A quantum  version of the BSS theorem was obtained by Buscemi  \cite{buscemi2012comparison}, and also by Shmaya \cite{shmaya2005comparison} for the so-called quantum information structures, consisting of a pair of Hilbert spaces and a bipartite state. This framework contains also quantum experiments, which are naturally identified with certain separable bipartite states.   The results of Shmaya  were
 reformulated for comparison of quantum channels by Chefles in \cite{chefles2009quantum}.

On the other hand, Matsumoto \cite{matsumoto2010randomization} introduces a natural generalization of classical decision problems to quantum ones, where sets of decisions with 
payoff (or loss) funtions are replaced by  Hilbert spaces with a family of positive payoff operators and decision rules 
 are given by channels. In this setting, 
 a quantum randomization criterion was proved, using the minimax theorem similarly as in the classical case (see e.g. \cite{strasser1985statistics}).

Existence of channels mapping a set of quantum states to another given set is an important problem also 
outside the theory of comparison of statistical experiments. Maybe the first result of this kind is by Alberti and Uhlmann \cite{aluh1980problem}, 
where pairs of qubit states were 
considered. More general situations were studied e.g. in \cite{cjw2004ontheexistence}, recently in \cite{hlps2012physical,hjrw2012extending}.

In the present work, we introduce the theory of comparison and randomization theorems  for quantum channels. First, we identify quantum statistical experiments with classical-to-quantum (cq-) channels and note that 
 randomizations correspond to compositions with channels (post-processings). We then show that the quantum randomization criterion
 proved in \cite{matsumoto2010randomization} can be formulated in terms of the diamond norm  and its dual 
 norm introduced and studied in \cite{gutoski2012measure, jencova2013base}. In fact, it turns out that it is a consequence of  duality of 
these norms.  

These concepts are
extended to the case when experiments, randomizations and decision rules are  replaced by general 
positive trace preserving maps (for simplicity also called channels), where positivity is given by some 
(family of) proper cones satisfying certain 
invariance properties.   We define post-processing deficiency of one channel with respect to the other and
prove a variant of the randomization theorem for post-processings.  
 We show the relation to pointwise deficiency, which is deficiency for pairs of experiments obtained by applying either channels to statistical experiments in the input algebra. 
If the maps are completely positive, we prove that the purely quantum deficiency is equivalent to classical deficiency
when the two channels are tensored by a suitable fixed channel, thus reformulating and extending the results of 
\cite{shmaya2005comparison} and \cite{buscemi2012comparison}. Combining this with pointwise deficiency, we get the formulation of Chefles.
A similar theory is obtained for pre-processings.

Quantum statistical  experiments (cq-channels) and quantum-to-classical (qc-) channels are discussed separately.
For experiments, we obtain the randomization criterion of \cite{matsumoto2010randomization}, but extended to the case when more general 
positive maps are considered. Perhaps the most important result here is putting the theory into the convenient 
 framework of base section norms and their duality. The qc-channels can be identified with positive operator valued measures (POVMs), that describe quantum mesurements. It is shown that pre- and post-processing deficiency is closely related to pre- and post-processing cleanness of the POVMs, defined in \cite{bkdpw2005clean}. 

Admissible families of cones in spaces of Hermitian linear maps are defined in the preliminary section \ref{sec:preliminary}, together with the corresponding pairs of norms $\|\cdot\|_\diamond$ and $\|\cdot\|^\diamond$. A general theory of base sections in partially ordered vector spaces and the corresponding norms is developed in the  Appendix.

\section{Preliminaries}\label{sec:preliminary}
Throughout the paper, all Hilbert spaces and C*-algebras are finite dimensional.
If $H$ is a Hilbert space, we denote  $d_H:=\dim(H)$ and fix an orthonormal basis $\{|i_H\>, i_H=1\dots d_H\}$  in $H$. 
We will denote the algebra of linear operators on $H$ by  $B(H)$, the set of positive operators in $B(H)$ by $B(H)^+$ and 
 the real vector space of Hermitian elements in $B(H)$ by $B_h(H)$. 
 
Any C*-algebra $\Ae$ will be represented as $\Ae=\bigoplus_k B(H_k)\subseteq B(H)$  for some Hilbert spaces $H_k$ and $H=\oplus_k H_k$.  In particular, an $n$-dimensional commutative C*-algebra will be identified with the algebra $\diag_n$ generated by  the projections $\{|i_H\>\<i_H|, i=1,\dots,n\}$ for some $n$-dimensional Hilbert space $H$. We denote $\Ae^+=\Ae\cap B(H)^+$ and  $\Ae_h=\Ae\cap B_h(H)$. Moreover, we denote by $E_\Ae$ the trace preserving conditional expectation of 
$B(H)$ onto $\Ae$, determined by
\[
\Tr ab=\Tr E_\Ae(a)b,\qquad a\in B(H),\ b\in \Ae.
\]

\subsection{Spaces of Hermitian maps}

Let $\Le(H,K)$ denote the real linear space of linear maps $B_h(H)\to B_h(K)$, then $\Le(H,K)$ can be identified with the
space  of 
Hermitian linear maps $B(H)\to B(K)$.
For $\phi\in \Le(H,H)$ we  define 
\[
s(\phi)=\sum_{i,j} \<i_H,\phi(|i_H\>\<j_H|)j_H\>.
\]
It is easy to see that $s$ defines a  linear functional $s: \Le(H,H)\to \mathbb R$. The next lemma shows that $s$ has tracelike properties with respect to composition of maps.
 
\begin{lemma}\label{lemma:s}
For all $\phi\in \Le(H,K)$, $\psi\in \Le(K,H)$, $s(\psi\circ\phi)=s(\phi\circ\psi)$.
\end{lemma}

\begin{proof} 
Let $\Phi_{b,a}$ denote the map $B(K)\ni x\mapsto (\Tr bx) a$, with $a\in B_h(H)$, $b\in B_h(K)$. We have
 \begin{align*}
s(\Phi_{b,a}\circ\phi)&=\sum_{i,j} \<i_H,aj_H\>\Tr b\phi(|i_H\>\<j_H|)=\sum_{i,j}\<i_H,aj_H\> \<j_H,\phi^*(b)i_H\>\\&=\Tr \phi(a)b,
 \end{align*}
 where $\phi^*$ is the adjoint with respect to Hilbert-Schmidt inner product:
\[
\Tr a\phi^*(b)=\Tr\phi(a)b,\qquad a\in B_h(H), b\in B_h(K).
\]
Similarly
\begin{align*}
s(\phi\circ\Phi_{b,a})&=\sum_{i,j}\<i_K,\phi(a)j_K\> \Tr b|i_K\>\<j_K|=\sum_{i,j}\<i_K,\phi(a)j_K\> \<j_K, b i_K\>\\
&=\Tr\phi(a)b=s(\Phi_{b,a}\circ\phi).
\end{align*}
Since the maps $\Phi_{b,a}$ generate $\Le(K,H)$, the statement follows.

\end{proof}

We now identify the dual space of $\Le(H,K)$ with $\Le(K,H)$, where duality is given by
\[
\<\phi,\psi\>=s(\phi\circ\psi),\qquad \phi\in \Le(H,K),\psi\in\Le(K,H).
\]
This duality is closely related to the inner product $\<\cdot,\cdot\>''$ in $\Le(H,K)$, introduced in \cite{skowronek2011cones}.
Note that the properties of $s$ imply that we have $\<\phi,\psi\>=\<\psi,\phi\>$ and 
\begin{equation}\label{eq:trian}
\<\alpha\circ\beta,\psi\>=s(\alpha\circ\beta\circ\psi)=\<\beta,\psi\circ\alpha\>=\<\alpha,\beta\circ\psi\>
\end{equation}
whenever $\alpha,\beta$ and $\psi$ are Hermitian maps with appropriate input and output spaces:
 \[
\xymatrix@C=2pc@R=2pc{B(H) \ar[r]^\beta & B(K) \ar[d]^\alpha \\
 & B(L)\ar[ul]_\psi}
\]
This simple observation is one of the most important tools used below. 

\begin{rem}\label{rem:choi}
Let us denote $X_H:=\sum_{i,j} |i_H\>\<j_H|\otimes |i_H\>\<j_H|$.  The \emph{Choi representation} \cite{choi1975completely}
$C: \phi\mapsto (\phi\otimes id_H)(X_H)$ provides an isomorphism of $\Le(H,K)$ onto $B_h(K\otimes H)$.
Note  that for any $\phi\in \Le(H,H)$, $s(\phi)=\Tr C(\phi)X_H$, so that for $\phi\in  \Le(H,K)$, 
$\psi\in \Le(K,H)$ we obtain
\begin{align}
\<\phi,\psi\>&=s(\psi\circ\phi)=\Tr C(\psi\circ\phi)X_H=\Tr (\psi\otimes id_H)(C(\phi))X_H\notag \\
&=\Tr C(\phi)(\psi^*\otimes id_H)(X_H)=\Tr C(\phi)C(\psi^*).\label{eq:chois}
\end{align}
It is of course possible to use this representation and we will do it in some places, but for our purposes it is mostly much more convenient to work with the spaces of mappings.

\end{rem}

Let now $\Ae$ and $\Be$ be C*-algebras and let $\Le(\Ae,\Be)$ denote the space of linear maps $\Ae_h\to\Be_h$. 
If $\Ae\subseteq B(H)$ and $\Be\subseteq B(K)$, then we may identify any $\phi\in \Le(\Ae,\Be)$ with 
$E_\Be\circ\phi\circ E_\Ae\in \Le(H,K)$. 
In this way, we may suppose that $\Le(\Ae,\Be)$ is a linear subspace in $\Le(H,K)$. In particular, we may define the 
functional $s$ on $\Le(\Ae,\Ae)$ by restriction. If $\phi\in \Le(\Ae,\Be)$ and $\psi\in \Le(K,H)$, then by (\ref{eq:trian})
\[
\<\phi,\psi\>=\<\phi, E_\Ae\circ\psi\circ E_\Be\>=\<\phi,\psi'\>
\]
 for $\psi'=E_\Ae\circ (\psi|_\Be)\in \Le(\Be,\Ae)$, so that we may identify the dual space of $\Le(\Ae,\Be)$ with $\Le(\Be,\Ae)\subseteq\Le(K,H)$. 
 
 \begin{ex}\label{ex:cqqc}
Let $d_H=n$, $d_K=m$. Any map $\phi\in \Le(\diag_n,B(K))$ has the form 
\[
\phi: |i_H\>\<i_H|\mapsto B_i,\qquad B=\{B_1,\dots,B_{n}\}\subseteq B_h(K).
\]
This can be extended to the map $\Phi^{cq}_{B}\in \Le(H,K)$, $a\mapsto \sum_i\<i_H,ai_H\> B_i$. If $B\subset B(K)^+$, 
then $\Phi^{cq}_B$ is called a \emph{classical-to-quantum (cq-) map}. Moreover, if  $B\subset \states(K)$, then $\Phi^{cq}_B$ is trace preserving, in this case it is called a \emph{cq-channel}.

Similarly, any map in $\Le(B(H),\diag_m)$ has the form 
\[
\Phi^{qc}_{A}: b\mapsto \sum_{i=1}^m (\Tr A_ib)|i_K\>\<i_K|,\qquad A=\{A_1,\dots,A_m\}\subseteq B_h(H).
\]
If $A\subset B(H)^+$, $\Phi^{qc}$ is called a \emph{quantum-to-classical (qc-) map}. If moreover $A$
is a POVM, that is, $A\subset B(H)^+$ and $\sum_i A_i=I$, $\Phi^{qc}_A$ is trace preserving and is called a \emph{qc-channel}.

If $\psi\in \Le(K,H)$, then 
$\Phi^{cq}_B\circ \psi=\Phi^{cq}_B\circ \Phi^{qc}_F=:\Phi_{F,B}$, where $ F=\{F_i:=\psi^*(|i_H\>\<i_H|), i=1,\dots,n\}$,
so that 
\begin{equation}\label{eq:s_cq}
\<\Phi^{cq}_B,\psi\>=\<\Phi^{cq}_B,\Phi^{qc}_F\>=s(\Phi_{F,B})=\sum_i \Tr B_iF_i
\end{equation}
Similarly, $\psi \circ \Phi^{qc}_A=\Phi_{A,G}$, $G=\{G_i:=\psi(|i_K\>\<i_K|), i=1,\dots,m\}$ and
\[
\<\Phi^{qc}_A,\psi\>=\<\Phi^{qc}_A,\Phi^{cq}_G\>=s(\Phi_{A,G})=\sum_i\Tr G_iA_i.
\]
 \end{ex}

\subsection{Admissible families of positive cones}

We now choose a convex cone in $\Le(H,K)$, that will serve as a positive cone. In fact, since we will work with mappings  with different input and 
output spaces, we have  to fix a family of cones, for all pairs of Hilbert spaces, and since we will consider compositions of maps, this family must satisfy certain invariance under composition. In the sequel, we will use definitions and facts that can be found in the  Appendix.

There is an obvious candidate for a positive cone in $\Le(H,K)$, namely the cone
$Pos(H,K)$ of all \emph{positive} maps in $\Le(H,K)$, that is, maps satisfying 
$\phi(B(H)^+)\subseteq B(K)^+$. One can see that this is a proper  cone in $\Le(H,K)$. Its dual with respect to our duality is the cone  $Pos(H,K)^*=EB(K,H)$ of \emph{superpositive} maps \cite{ando2004cones}, or  \emph{entanglement breaking} maps \cite{hsr2003entanglement}, that is, maps of the form
\[
\Phi_{F,\Ee}: a\mapsto \sum_i \rho_i\Tr F_ia,
\]
where  $\Ee=\{\rho_1,\dots,\rho_k\}\subset B(H)^+$ and $F=\{F_1,\dots,F_k\}\subset B(K)^+$.
 In particular, if $F_i=|i_K\>\<i_K|$ for $i=1,\dots,d_K=k$,
then $\Phi_{\Ee,F}=\Phi^{cq}_\Ee$  is a cq-map, see Example \ref{ex:cqqc}.   Similarly, if
$\rho_i=|i_H\>\<i_H|$ for $i=1,\dots,d_H=k$, then $\Phi_{\Ee,F}=\Phi^{qc}_F$  is a qc-map.
Note that any EB map is a composition of a qc- and a cq-map, $\Phi_{\Ee,F}=\Phi^{cq}_\Ee\circ\Phi^{qc}_F$.

\begin{rem}\label{rem:EB} Let $\Phi_{F,\Ee}$ be an EB map. Note that we may always rearrange the operators in $\Ee$ and $F$ such that 
$\Ee\subset \states(H)$ or such that $F$ is a POVM. To see the latter,  let $t=\|\sum_i F_i\|$ and put
$F'_i=t^{-1}F_i$, $i=1,\dots,k$, $F'_{k+1}=I-\sum_{i=1}^k F_i'$,  $\rho'_i=t\rho_i$, $i=1,\dots,k$, $\rho'_{k+1}=0$. Then $F'$ is a POVM and  
 $\Phi_{F,\Ee}=\Phi_{F',\Ee'}$. Note also that we may suppose both, so that $\Phi_{F,\Ee}$ is a composition of a cq- and a 
qc-channel, if and only if $\Phi_{F,\Ee}$ is trace preserving.

\end{rem}

Another natural choice is the cone  $CP(H,K)$ of \emph{completely positive} maps in $\Le(H,K)$, that is, maps 
satisfying $\phi\otimes id_F\in Pos(H\otimes F,K\otimes F)$ for all finite dimensional Hilbert spaces $F$.
In this case, $CP(H,K)^*=CP(K,H)$. Note also that we have $EB(H,K)\subseteq CP(H,K)\subseteq Pos(H,K)$ and the inclusions are proper if both spaces have dimension at least 2.

\begin{defi} Let $\Pe=\{\Pe(H,K)\}$ be a family of proper cones, where $\Pe(H,K)\subset \Le(H,K)$ and $H,K$ runs over all pairs of finite dimensional Hilbert spaces. We will say that $\Pe$ is an \emph{admissible family of positive cones} if the following properties are satisfied:
\begin{enumerate}
\item[(i)] $EB\subseteq \Pe\subseteq Pos$. 
\item [(ii)] $\Pe$ is closed  under composition:  $\Pe\circ\Pe\subseteq \Pe$.
\item[(iii)] $\Pe$ is closed under composition  with $CP$ maps: $CP\circ\Pe\circ CP\subseteq \Pe$.
\end{enumerate}

\end{defi}

It is easy to see that $EB$, $Pos$ and $CP$ are admissible families of positive cones.
Closed convex cones satisfying (iii) are sometimes  called \emph{mapping cones} or \emph{cones with a mapping cone symmetry}, see 
\cite{stormer1986extension, skowronek2011cones,ssz2009cones} for more results and 
examples. We give some further examples of admissible families.

\begin{ex}
Let $CP_k(H,K)$  be the cone of \emph{$k$-positive} maps, satisfying 
$\phi\otimes id_F\in Pos(H\otimes F,K\otimes F)$ if $d_F=k$. Obviously $CP\subseteq CP_k\subseteq CP_1=Pos$ and 
$CP_k(H,K)=CP(HK)$ if $k=\min\{d_H,d_K\}$. The dual space $CP_k(H,K)^*=EB(K,H)$ is the cone of \emph{$k$-entanglement breaking} maps \cite{ssz2009cones}, that is, maps such that $\phi\otimes id(\rho)$ has Schmidt rank less than $k$ for all bipartite states $\rho$.  It is clear that $EB_k$ is an admissible family as well.

\end{ex}

 By taking the dual cones, we obtain another family $\{\Pe^*(H,K)\}$  of proper  cones, such that
 $\Pe(H,K)^*=\Pe^*(K,H)$.  In all examples considered so far, both $\Pe^*$ and $\Pe$ were admissible families. In general,  it does not seem to be true that $\Pe^*$ is admissible if $\Pe$ is. However, $\Pe^*$ satisfies the properties (i) and 
(iii) and $\Pe^*\circ\Pe,\Pe\circ\Pe^*\subset \Pe^*\cap CP$  whenever $\Pe$ is admissible.

\begin{rem}\label{rem:norms}
The Choi isomorphism maps $\Pe(H,K)$ onto a proper cone in $B_h(K\otimes H)$. For example, $C(CP(H,K))=B(K\otimes H)^+$
 and  $C(EB(H,K))=Sep(K\otimes H)$, where
\[ 
Sep(K\otimes H):=\{\sum_i A_i\otimes B_i, A_i\in B(K)^+, B_i\in B(K)^+\}
\]
is the cone of positive separable elements. If duality in $B_h(K\otimes H)$ is given by 
$\<A,B\>=\Tr AB$, $A,B\in B_h(K\otimes H)$, the dual cone of $C(\Pe)$ is $C(\Pe)^*=C(\tilde \Pe)$,
 where $\tilde \Pe(H,K)=\{\phi^*, \phi\in \Pe^*(K,H)\}$, this follows from Remark \ref{rem:choi}. Note that we have $\Pe^*=\tilde \Pe$ in all examples. By property (iii), both $C(\Pe)$ and $C(\tilde \Pe)$ are invariant under maps of the form 
$\phi\otimes \psi$ where $\phi,\psi\in CP$, in particular, the cones are invariant under conjugation by local unitaries.   Let us denote the  order  in $B_h(K\otimes H)$ given by $C(\tilde \Pe)$ by $\le_{\tilde \Pe}$. 

Let $\tau=\tau_{H,K}:=\Phi_{I,I}:  a\mapsto (\Tr a)I_K$. Then $\tau$ is an interior point in the cone $EB(H,K)$ and by the property (i) it is easy to see that $\tau$ is an order unit for all cones $\Pe(H,K)$, $\Pe^*(H,K)$ and $\tilde \Pe(H,K)$.
Consequently, $C(\tau)=I$ is an order unit with respect to the order $\le_{\tilde \Pe}$ in $B_h(K\otimes H)$.
In the sequel, we will need the corresponding order unit norm, given by
\[
\|X\|_{\tilde \Pe}:=\inf\{\lambda>0, -\lambda I\le_{\tilde\Pe} X\le_{\tilde \Pe} \lambda I\}
\]
and its dual, the base norm
\[
\|X\|_{1,\tilde \Pe}:=\sup_{-I\le_{\tilde \Pe}Y\le_{\tilde \Pe}I} \Tr XY=2\sup_{0\le_{\tilde \Pe} Y\le_{\tilde \Pe}I}\Tr XY-\Tr X.
\]
Note that for $\Pe=CP$, we have $\tilde \Pe=CP$, $\|\cdot\|_{CP}$ is the operator norm and  $\|\cdot\|_{1,CP}=\|\cdot\|_1$ is the 
$L_1$-norm in $B(K\otimes H)$. If $\Pe=Pos$, then $\tilde \Pe=Sep$ and $\|\cdot\|_{1,Sep}$ is the norm $\|\cdot\|_{\mathbf{SEP}}$, 
defined  e.g. in \cite{mww2009distinguishability}.

\end{rem}

If $\Ae\subseteq B(H)$ and $\Be\subseteq B(K)$ are C*-algebras, then $\Pe(\Ae,\Be):=\Pe(H,K)\cap \Le(\Ae,\Be)$ is a proper cone and its dual is $\Pe^*(\Be,\Ae)=\Pe^*(K,H)\cap \Le(\Be,\Ae)$. In this way, we obtain a family of cones $\{\Pe(\Ae,\Be)\}$ for all pairs of C*-algebras, satisfying the admissibility properties (i)-(iii).

\subsection{The diamond norm and its dual}

From now on, we fix an admissible family of positive cones $\Pe$. It is clear that many of the following notions depend on this choice and we will sometimes use the subscript $\Pe$ to emphasize this dependence.
We will denote the corresponding order in $\Le(H,K)$ 
 by $\le$ and the dual order in $\Le(K,H)$ by $\le^*$. Let us also denote  the set of trace preserving elements in 
$\Pe(H,K)$ by
 $\Ce=\Ce(H,K)$  The elements of
$\Ce$ will be called \emph{channels} in the sequel, although usually in the literature channels are required to be 
completely positive. We will use the expression \emph{$\Pe$-channel} if necessary. We denote 
$\Ce(\Ae,\Be)=\Ce(H,K)\cap \Le(\Ae,\Be)$ for C*-subalgebras $\Ae\subseteq B(H)$, $\Be\subseteq B(K)$. Again, we refer to the Appendix for the definition of base sections and the corresponding norms.

\begin{lemma}\label{lemma:base} $\Ce(\Ae,\Be)$ is a base section in $(\Le(\Ae,\Be),\Pe(\Ae,\Be))$. The dual section is 
$\Se(\Be,\Ae):=\{\Phi_{I_K,\sigma}: a\mapsto (\Tr a)\sigma,\ \sigma\in \states(\Ae) \}$.
\end{lemma}

\begin{proof} We first note that  $\Se(\Be,\Ae)\subseteq EB(\Be,\Ae)\subseteq\Pe^*(\Be,\Ae)$ and $\Se(\Be,\Ae)$ contains
 the interior point $d_H^{-1}\tau_{K,H}$ of $\Pe^*(\Be,\Ae)$, see Remark \ref{rem:norms}.
   It is not difficult to see that $\Se(\Be,\Ae)$ is a base section in $(\Le(\Be,\Ae),\Pe^*(\Be,\Ae))$.  Since $\<\phi,\Phi_{I,\sigma}\>=\Tr \phi(\sigma)$ for any $\phi\in \Le(\Ae,\Be)$, it is easy to see that the dual section is $\Ce(\Ae,\Be)$. The rest of the proof  follows from Lemma \ref{lemma:app_dual} in the Appendix.

\end{proof}

We will denote the norm corresponding to $\Ce(\Ae,\Be)$ by $\|\cdot\|_\diamond$ and its dual  by $\|\cdot\|^\diamond$.
It was already shown in \cite{jencova2013base} that if $\Pe=CP$, $\|\cdot\|_{\diamond}$ is the  diamond norm considered in \cite{kitaev1997quantum, 
watrous2005notes}.
By Corollary \ref{coro:app_norm} (i) (in the Appendix), the two norms have the following form:
\begin{align}
\|\phi\|_{\diamond} &= \sup_{\sigma\in \states(\Ae)} \sup_{-\Phi_{I,\sigma}\le^*\xi\le^*\Phi_{I,\sigma}}
\<\phi,\xi\>\label{eq:diamond_sup}\\
&=
\inf_{\alpha\in \Ce(\Ae,\Be)}\inf \{\lambda>0, -\lambda\alpha\le\phi\le\lambda\alpha\}\label{eq:diamond_inf} \\
\|\phi\|^{\diamond}&=\sup_{\alpha\in \Ce(\Be,\Ae)} \sup_{-\alpha\le\xi\le\alpha}\<\phi,\xi\>
\label{eq:diamondu_sup}\\
&=
\inf_{\sigma\in \states(\Be)}\inf\{\lambda>0, -\lambda\Phi_{I,\sigma}\le^* \phi\le^* \lambda\Phi_{I,\sigma}\}
\label{eq:diamondu_inf}
\end{align}
for all $\phi\in \Le(\Ae,\Be)$. We will now find a more explicit form of $\|\cdot\|_\diamond$.

\begin{thm}\label{thm:diamond} For any $\phi\in \Le(H,K)$,
\[ 
\|\phi\|_{\diamond,\Pe}=\sup_{\rho\in \states({H\otimes H})}\|(\phi\otimes id_H)(\rho)\|_{1,\tilde \Pe}.
\]

\end{thm}

\begin{proof} Let $\phi\in \Le(H,K)$. Using Remarks \ref{rem:choi}, \ref{rem:norms} and (\ref{eq:diamond_sup}), we obtain 
\begin{align*}
\|\phi\|_\diamond&=\sup_{\sigma\in \states(H)}\, \sup_{-\Phi_{I,\sigma}\le^*\xi\le^*\Phi_{I,\sigma}}\Tr C(\phi)C(\xi^*)\\
&=\sup_{\sigma\in \states(H)}\, \sup_{-(I\otimes \sigma)\le_{\tilde \Pe} X\le_{\tilde \Pe}(I\otimes \sigma)}\Tr C(\phi)X=\\
&=\sup_{\sigma\in \states(H)} \sup_{-I\le_{\tilde \Pe} Y\le_{\tilde \Pe}I} \Tr C(\phi)(I\otimes\sigma^{1/2})Y(I\otimes\sigma^{1/2})\\
&=\sup_{\sigma\in \states(H)}\|(I\otimes \sigma^{1/2})C(\phi)(I\otimes\sigma^{1/2})\|_{1,\tilde \Pe}\\
&=\sup_{\sigma\in \states(H)} \|(\phi\otimes id_H)(|x_\sigma\>\<x_\sigma|)\|_{1,\tilde \Pe},
\end{align*}
where $|x_\sigma\>=\sum |i_H\>\otimes \sigma^{1/2}|i_H\>$ is a unit vector in $H\otimes H$.
On the other hand, let $|x\>\in H$ be a unit vector, then $|x\>=\sum |i_H\>\otimes|x_i\>$ for
 some $|x_i\>\in H$. Let $R: H\to H$
 be a linear map given by $R|i_H\>=|x_i\>$, then $R^*R=:\sigma\in \states(H)$ and $|x_i\>=U\sigma^{1/2}|i_H\>$ 
 for some unitary $U\in B(H)$. It is easy to see  that $\|\cdot\|_{1,\tilde \Pe}$ is invariant under local unitaries, so that
\begin{align*}
\|(\phi\otimes id_H)(|x\>\<x|)\|_{1,\tilde \Pe}&=\|(I\otimes U)(\phi\otimes id_H)(|x_\sigma\>\<x_\sigma|)(I\otimes U^*)\|_{1,\tilde \Pe}\\
&=\|(\phi\otimes id_H)(|x_\sigma\>\<x_\sigma|)\|_{1,\tilde \Pe}
\end{align*}
The proof now follows by convexity of the norm $\|\cdot\|_{1,\tilde \Pe}$.

\end{proof}

Using the Choi representation, we also obtain an expression for the dual norm. For  $a\in B(H)^+$ and $x\in B_h(H)$, let
$\|a^{-1/2}xa^{-1/2}\|:=\lim_{\varepsilon\to 0 } \|(a+\varepsilon I)^{-1/2}x(a+\varepsilon I)^{-1/2}\|$. Note that the limit is finite if $\mathrm{supp} (x)\le \mathrm{supp}(a)$, where $\mathrm{supp}(x)$ is the projection onto the range of $x$, otherwise $\|a^{-1/2}xa^{-1/2}\|=\infty$.

\begin{lemma}\label{lemma:cp_udiam} Let $\phi\in \Le(H,K)$, then 
\[
\|\phi\|^\diamond_{\Pe}=\inf_{\sigma\in\states(K)}\|(I_H\otimes \sigma^{-1/2})C(\phi^*)(I_H\otimes \sigma^{-1/2})\|_{\tilde \Pe}.
\]

\end{lemma}

\begin{proof} By (\ref{eq:diamondu_inf}) and Remark \ref{rem:norms}, we obtain
\begin{align*}
\|\phi\|^\diamond_{\Pe}&=\inf_{\sigma\in \states(K)}\inf\{\lambda>0, C(\lambda \Phi_{I,\sigma}\pm \phi)\in C(\Pe^*)\}\\
&=\inf_{\sigma\in \states(K)}\inf\{\lambda>0, C(\lambda \Phi_{I,\sigma}^*\pm\phi^*)\in C(\tilde \Pe)\}
\end{align*}
Since $C(\Phi^*_{I,\sigma})=I_H\otimes \sigma^t$, where $\sigma^t$ is the transpose of $\sigma$, this proves the lemma.

\end{proof}

An important property of the norms is monotonicity under composition with channels.

\begin{lemma}\label{lemma:monot}
Let $\alpha$ and $\beta$ be $\Pe$-channels or CP-channels, $\phi\in \Le(\Ae,\Be)$. Then
\[
\|\beta\circ\phi\circ\alpha\|_\diamond\le \|\phi\|_\diamond,\qquad \|\beta\circ\phi\circ\alpha\|^\diamond\le \|\phi\|^\diamond.
\]

\end{lemma}

\begin{proof} Consider the maps $R_\alpha: \phi\mapsto\phi\circ\alpha$, $L_\beta: \phi\mapsto \beta\circ\phi$. Then $L_\beta$ is a 
linear map $\Le(\Ae,\Be)\to \Le(\Ae,\Be')$, preserving the cone $\Pe$ and  such that  $L_\beta(\Ce)\subseteq \Ce$, $L_\beta(\Se)\subseteq \Se$, similarly for  $R_\alpha$. The proof now follows by  Proposition \ref{prop:app_maps} (i), in the Appendix.
 \end{proof}

We now compute the norms  in the special case of the cq- and qc- maps. Note that all such maps belong to 
$\Le(\diag_n, B(H))$ resp. $\Le(B(H),\diag_n)$ and in this case $Pos=EB$, so that the norms do not depend from the choice of $\Pe$.

\begin{lemma}\label{lemma:qccq}  Let  $A=\{A_1,\dots,A_n\}\subset \Ae_h$. Then we have 
\begin{enumerate}
\item[(i)] $\|\Phi^{cq}_A\|^\diamond=\inf_{\sigma\in \states(\Ae)} \max_i\|\sigma^{-1/2}A_i\sigma^{-1/2}\|$.
\item[(ii)] $\|\Phi^{qc}_A\|^\diamond=\sum_i \|A_i\|$.
\item[(iii)] $\|\Phi^{cq}_A\|_{\diamond}=\max_i \|A_i\|_1$.
\item[(iv)] $\|\Phi^{qc}_A\|_{\diamond}=\sup_{\sigma\in \states(\Ae)}\sum_i\|\sigma^{1/2}A_i\sigma^{1/2}\|_1$
\end{enumerate}

\end{lemma}

\begin{proof} 
\begin{enumerate}
\item[(i)] For $\sigma\in \states(\Ae)$ and $\lambda>0$, $-\lambda\Phi_{I,\sigma}\le^*\Phi^{cq}_A\le^*\lambda\Phi_{I,\sigma}$ if and only if $-\lambda\sigma\le A_i\le \lambda\sigma$ for all $i$.
If $\mathrm{supp} (A_i)\subseteq \mathrm{supp}(\sigma)$, then the infimum over  such $\lambda$ is equal to 
$\max_i\|\sigma^{-1/2}A_i\sigma^{-1/2}\|$, otherwise the infimum is $\infty$. By  (\ref{eq:diamondu_inf}), this implies
 (i).
\item[(ii)] Any $\alpha\in \Ce(\diag_n,\Ae)$ has the form $\alpha=\Phi^{cq}_\Ee$, $\Ee=\{\rho_1,\dots,\rho_n\}\subset \states(\Ae)$ and by (\ref{eq:diamondu_sup}),
\[
\|\Phi_A^{qc}\|^\diamond=\sup_{\substack{\rho_i\in \states(\Ae)\\-\rho_i\le F_i\le \rho_i, \forall i}}\sum_i\Tr A_iF_i=
\sum_i \sup_{\substack{\rho\in\states(\Ae)\\ -\rho\le F\le \rho}} \Tr A_iF=\sum_i\|A_i\|.
\]
\item[(iii)] Using (ii) and duality of the norms, we obtain
\[
\|\Phi^{cq}_A\|_\diamond=\sup_{\substack{B_i\in \Ae_h\\,\sum_i\|B_i\|\le 1}} \Tr A_iB_i=\max_i\|A_i\|_1.
\]
\item[(iv)]  Follows similarly  from (i) by duality. 
\end{enumerate}
 
\end{proof}

\section{Randomization theorems}
In this section we generalize the theory of deficiency and randomization  of  quantum statistical experiments to quantum channels.

\subsection{Comparison of statistical experiments}\label{sec:comparison}

In classical statistics, the theory of statistical experiments and their comparison was introduced by Blackwell in \cite{blackwell1951comparison} and further developed by many authors, for more information see \cite{torgersen1991comparison} or \cite{strasser1985statistics}. This theory has been extended to the quantum case in \cite{buscemi2012comparison,matsumoto2010randomization,jencova2012comparison}.

  Although the theory e.g. in \cite{matsumoto2010randomization} allows a more general setting, we will define a \emph{(quantum) statistical experiment} (or just \emph{experiment}) as a finite subset $\Ee=\{\rho_1,\dots,\rho_n\}\subset \states(\Ae)$ for some finite dimensional C*-algebra $\Ae$.
Note that this setting contains also classical statistical experiments supported on a finite set $X$, $|X|=N$, by putting  $\Ae=\diag_N$. 

Any experiment can be 
viewed as the set of possible states of some physical system, determined by some prior information on the true state.
 Based on the outcome of a measurement on the system, a decision $j$  is chosen from a (finite) set $D$ of decisions. This procedure, or a \emph{decision rule}, is represented by a POVM with outcomes in $D$, that is, a
collection $M=\{M_j, j\in D\}$ of positive operators in $\Ae$ such that $\sum_j M_j=I$. We will denote the set of all 
POVMs with $k$ outcomes  by $\povm(\Ae,k)$ and the set of all POVMs in $\Ae$ by $\povm(\Ae)$. 

The performance of a decision rule is assessed by a \emph{payoff function}, which in our case is an $n\times k$ matrix 
$g=(g_{ij})$ with nonnegative entries representing the {payoff} 
obtained if $j\in D$ is chosen while the true state is $\rho_i$.  The average payoff of the decision rule $M$ at $\rho_i$ is computed as 
\[
P_\Ee(i,M,g)=\sum_{j\in D} g_{ij}\Tr\rho_iM_j.
\]
We call the pair $(D,g)$ a \emph{(classical) decision space}.  The triple $(\Ee,D,g)$ where $\Ee$ is an experiment and $(D,g)$ is a decision space is called a 
\emph{(classical) decision problem}. 

As a natural generalization, a \emph{quantum  decision space} was defined in \cite{matsumoto2010randomization}
 as a Hilbert space $D$ with a collection of positive payoff operators $G=\{G_1,\dots,G_n\}\subset B(D)^+$, or more generally a C*-algebra $\De$ with
$G\subset \De^+$.  In this case, decision rules  are represented by CP-channels 
$\alpha: \Ae\to \De$ and the payoff is computed as 
\[
P_\Ee(i,\alpha,G)=\Tr \alpha(\rho_i)G_i.
\] 
It is easy to see that classical decision spaces can be identified with quantum ones with $\De=\diag_k$  and that in  this case, any decision rule is given by  a POVM. 

Similarly as for classical statistical experiments, deficiency of one experiment with respect to another is defined by comparing all possible payoffs of decision rules for all decision spaces. The following definition is very similar to the one given in 
\cite{matsumoto2010randomization}. 

\begin{defi}\label{def:defi_exps}
Let $\Ee=\{\rho_1,\dots,\rho_n\}\subset \states(\Ae)$, $\Fe=\{\sigma_1,\dots,\sigma_n\}\subset\states(\Ae')$ be statistical experiments. 
Let $\epsilon\ge 0$ and let $\De$ be a C*-algebra.
We say that $\Fe$ is \emph{ $(\epsilon,\De)$-deficient} with respect
 to $\Ee$, in notation $\Fe\succeq_{\epsilon,\De}\Ee$, if for every sequence $G=\{G_1,\dots,G_n\}\subset \De^+$
 and any $\alpha\in \Ce_{CP}(\Ae,\De)$ there is some $\alpha'\in \Ce_{CP}( \Ae',\De) $ such that
\[
P_\Ee(i,\alpha,G)\le P_\Fe(i,\alpha',G)+\epsilon\|G_i\|,\qquad  i=1,\dots,n.
\]
If $\Fe\succeq_{\epsilon,\De}\Ee$ for all $\De$, we say that $\Fe$ is \emph{$\epsilon$-deficient} with respect to $\Ee$,
 in notation $\Fe\succeq_\epsilon \Ee$.
\end{defi}

If we restrict to $\De=\diag_k$, equivalently to decision spaces with commuting payoff operators, the corresponding
$(\epsilon,\De)$-deficiency is called \emph{classical}.

Let $\alpha\in \Ce_{CP}(\Ae,\Be)$ and let $\Ee\subset \states(\Ae)$ be an experiment. Then $\alpha(\Ee)\subset\states(\Be)$ is again an experiment, called a \emph{randomization} of $\Ee$. 
The classical randomization theorem proved by Torgersen \cite{torgersen1970comparison}, see also \cite{torgersen1991comparison} or \cite{strasser1985statistics}, relates (classical) deficiency
of classical statistical experiments to their randomizations.
 The following  quantum generalization of the theorem was obtained in  
\cite[Theorem 4]{matsumoto2010randomization}, for classical decision spaces  also in \cite[Theorem 2]{jencova2012comparison}. 

\begin{thm}
\label{thm:qrand} Let $\Ee=\{\rho_1,\dots,\rho_n\}\subset\states(\Ae)$,
 $\Fe=\{\sigma_1,\dots,\sigma_n\}\subset\states(\Ae')$. Then $\Fe\succeq_{\epsilon,\De}\Ee$ if and only if for every $\alpha\in \Ce_{CP}(\Ae,\De)$ there is some $\alpha'\in \Ce_{CP}(\Ae',\De)$ such that 
\begin{equation}\label{eq:qrt}
\|\alpha'(\sigma_i)-\alpha(\rho_i)\|_1\le 2\epsilon,\quad i=1,\dots,n
\end{equation}
and  $\Fe\succeq_\epsilon\Ee$ if and only if there is a randomization $\Fe'=\{\sigma_1',\dots,\sigma_n'\}\subset \states(\Ae)$ of 
$\Fe$, such that $\|\sigma_i'-\rho_i\|_1\le 2\epsilon$ for all $i$.

\end{thm}

We now make a few observations. First, for any experiment $\Ee\subset \states(\Ae)$, consider   the cq-channel 
$\Phi^{cq}_\Ee:\diag_n\to \Ae$. If $\alpha\in \Ce_{CP}(\Ae,\Be)$, then $\Phi^{cq}_{\alpha(\Ee)}=\alpha\circ\Phi^{cq}_\Ee$, so that randomizations of the experiment are obtained by (post-) composition of the cq-channel. 
Furthermore, let $(\De,G)$ be a decision space, then we may identify the collection of payoff operators $G$ with the 
qc-channel $\Phi^{qc}_G:\De\to \diag_n$. Define
\[
P_\Ee(\alpha,G):=\<\alpha\circ\Phi^{cq}_\Ee,\Phi^{qc}_G\>=\sum_i P_\Ee(i,\alpha,G).
\]
It can be shown, e.g. by using the Hahn-Banach separation theorem as in \cite[Proposition 1]{buscemi2012comparison}, that  in Definition \ref{def:defi_exps}, the inequalities for each $i$ can be replaced by their sum over $i$, that is,  by
\begin{equation}\label{eq:defic}
P_\Ee(\alpha,G)\le P_\Fe(\alpha',G)+\epsilon\sum_i\|G_i\|=P_\Fe(\alpha',G)+\epsilon\|\Phi^{qc}_G\|^\diamond,
\end{equation}
where the last equality follows by Lemma \ref{lemma:qccq}. By the same Lemma, we see that  (\ref{eq:qrt}) is equivalent to
\[
\|\alpha\circ\Phi^{cq}_\Ee-\alpha'\circ\Phi^{cq}_\Fe\|_\diamond\le 2\epsilon.
\]

It is now clear that deficiency of quantum experiments and the quantum randomization theorem can be reformulated in terms of the dual norms 
$\|\cdot\|_\diamond$ and  $\|\cdot\|^\diamond$. 

\subsection{Post-processing deficiency}

We will now suppose that an admissible family $\Pe$ of positive cones is fixed. Let $\Phi\in \Ce(\Ae,\Be)$ and let 
$\De$ be a C*-algebra, 
which will be interpreted  as an algebra  of decisions. As before, decision rules will be given by channels $\alpha\in \Ce(\Be,\De)$, which 
 will act on the channel $\Phi$ by post-composition. In this way, we obtain from $\Phi$ a $\Pe$-channel $\alpha\circ\Phi: \Ae\to \De$, this is called a \emph{post-processing} of $\Phi$.
 
 Next we define the payoff connected with post-processings. In general, we may take a positive affine functional $\Gamma$  
 on the set $\Ce(\Ae,\De)$, that will serve  as the payoff functional,  and define the  payoff by  $P_\Phi(\alpha,\Gamma):=\Gamma(\alpha\circ \Phi)$. 
 Since $\Ce(\Ae,\De)$ is a base section in 
 $\Le(\Ae,\De)$, by Lemma \ref{lemma:app_krein} (Appendix) there is an element in $\Pe^*(\De,\Ae)$, also denoted by $\Gamma$,
 such that $P_\Phi(\alpha,\Gamma)=\<\alpha\circ\Phi,\Gamma\>$. Thus we define a  \emph{post-processing decision space } 
 as a pair $(\De,\Gamma)$, where $\Gamma\in \Pe^*(\De,\Ae)$ is called a \emph{payoff map}.  A \emph{post-processing decision problem} 
 is a triple $(\Phi,\De,\Gamma)$ with $\Phi\in \Ce(\Ae,\Be)$ and $\Gamma\in \Pe^*(\De,\Ae)$.

\begin{defi}
 Let $\Phi\in \Ce(\Ae,\Be)$ and $\Psi\in \Ce(\Ae,\Be')$. Let $\epsilon \ge 0$ and let $\De$ be a  C*-algebra.
 We say that $\Psi$ is \emph{$(\epsilon,\De)$-post-processing deficient}  with respect to $\Phi$, in notation  
$\Psi\succeq_{\epsilon,\De}\Phi$, if for any  $\Gamma\in \Pe^*(\De,\Ae)$ 
and any $\alpha\in \Ce(\Be,\De)$, there is some $\alpha'\in \Ce(\Be',\De)$ such that
 \[
\<\alpha\circ\Phi,\Gamma\>\le \<\alpha'\circ\Psi,\Gamma\>+\epsilon \|\Gamma\|^\diamond.
\]
If $\Psi\succeq_{\epsilon,\De}\Phi$ for all $\De$, we say that $\Psi$ is \emph{ $\epsilon$-post-processing deficient} with respect to $\Phi$, in notation, $\Psi\succeq_{\epsilon}\Phi$.
\end{defi}

We visualise the maps in the above definition in the following diagram, where the solid arrows represent the channels and the decision rules, while  the dashed arrow represents the payoff map:
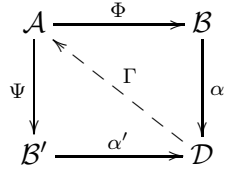
\[
\xymatrix@C=4pc@R=3pc{\Ae \ar[r]^\Phi \ar[d]_\Psi & \Be \ar[d]^\alpha\\
\Be'\ar[r]^{\alpha'} & \De \ar@{-->}[ul]_\Gamma}
\]
The randomization theorem for post-processings states that $\Psi\succeq_{\epsilon,\De}$ if and only if  for all  $\alpha$ we may choose 
$\alpha'$  such that the diagram (of solid arrows) commutes up to  $\epsilon$ in the $\|\cdot\|_\diamond$-norm.

\begin{thm}\label{thm:post}\textbf{(Post-processings)} The following are equivalent.
\begin{enumerate}
\item[(i)] $\Psi\succeq_{\epsilon,\De}\Phi$.
\item [(ii)] For any  $\Gamma\in \Pe^*(\De,\Ae))$, 
\[
\|\Phi\circ\Gamma\|^\diamond\le \|\Psi\circ\Gamma\|^\diamond+ \epsilon \|\Gamma\|^\diamond
\]
\item [(iii)] For any $\Gamma\in \Le(\De,\Ae)$ and all $\alpha\in \Ce(\Be,\De)$, there is some $\alpha'\in \Ce(\Be',\De)$ such that
\[
\<\alpha\circ\Phi,\Gamma\>\le \<\alpha'\circ\Psi,\Gamma\>+2\epsilon \|\Gamma\|^\diamond
\]

\item [(iv)] For any $\alpha\in \Ce(\Be,\De)$ there is some $\alpha'\in \Ce(\Be',\De)$ such that 
\[
\|\alpha\circ\Phi-\alpha'\circ\Psi\|_{\diamond}\le 2\epsilon .
\]
\end{enumerate}
Moreover, if $id \in \Pe$, then $\Psi\succeq_{\epsilon}\Phi$ if and only if $\Psi\succeq_{\epsilon,\Be}\Phi$.
\end{thm}

\begin{proof}
Suppose (i), and let $\Gamma \in \Pe^*(\De,\Ae)$. Using Corollary \ref{coro:app_norm} (ii) (Appendix), we have for any $\phi\in \Ce(\Be,\De)$,
\begin{align*}
\<\phi,\Phi\circ\Gamma\>&=\<\phi\circ\Phi,\Gamma\>\le 
\sup_{\phi'\in \Ce(\Be',\De)} \<\phi'\circ\Psi,\Gamma\>+\epsilon\|\Gamma\|^\diamond\\
&=\sup_{\phi'\in \Ce(\Be',\De)}\<\phi',\Psi\circ\Gamma\>+\epsilon \|\Gamma\|^\diamond=
\|\Psi\circ\Gamma\|^\diamond+\epsilon \|\Gamma\|^\diamond.
\end{align*}
This implies (ii) and it is also not difficult to see that (ii) implies (i).

Again, suppose (i) is true and this time let $\Gamma\in \Le(\De,\Ae)$, $\|\Gamma\|^\diamond=t$. Then by (\ref{eq:diamondu_inf}) 
 there is some $\sigma\in \states(\Ae)$ such that $\Gamma+t\Phi^{cq}_\sigma\in \Pe^*(\De,\Ae)$ and 
$\|\Gamma+t\Phi^{cq}_\sigma\|^\diamond\le 2t$. By (i), for any $\alpha\in \Ce(\Be,\De)$, there is some $\alpha'\in \Ce(\Be',\De)$ such that
\[
\<\alpha\circ\Phi,\Gamma+t\Phi^{cq}_\sigma\>\le\<\alpha'\circ\Psi,\Gamma+t\Phi^{cq}_\sigma\>+\epsilon 2t 
\] 
Since $\<\Omega,t\Phi^{cq}_\sigma\>=t$ for any channel $\Omega$, this implies (iii).
  
Suppose (iii) and  let $\alpha\in \Ce(\Be,\De)$. Let $\Oe^\diamond$ be the unit ball for $\|\cdot\|^\diamond$. Then 
\[
\max_{\Gamma\in \Oe^\diamond}\min_{\alpha'\in \Ce(\Be',\De)}\<\alpha\circ\Phi-\alpha'\circ\Psi,\Gamma\>\le 2\epsilon.
\]
Since the sets $\mathcal O^\diamond$ and $\Ce(\Be',\De)$ are both compact and convex and the function
$(\Gamma,\alpha')\mapsto \<\alpha\circ\Phi-\alpha'\circ\Psi,\Gamma\>$ is linear in both arguments, the minimax theorem 
 applies, see e.g. \cite{strasser1985statistics}.
 We obtain
\[
2\epsilon\ge \min_{\alpha'\in \Ce(\Be',\De)}\max_{\Gamma\in \mathcal O^\diamond}\<\alpha\circ\Phi-\alpha'\circ\Psi,\Gamma\>=
\min_{\alpha'\in \Ce(\Be',\De)}\|\alpha\circ\Phi-\alpha'\circ\Psi\|_\diamond,
\]
which is (iv).

 Suppose (iv) and let $\Gamma\in \Pe^*(\De,\Ae)$, $\alpha\in \Ce(\Be,\De)$. Let $\alpha'\in \Ce(\Be',\De)$ be such that
 $\|\alpha\circ\Phi-\alpha'\circ\Psi\|_\diamond\le 2\epsilon $. Then by Corollary  \ref{coro:app_norm} (iii), we obtain that
\[
\<\alpha\circ\Phi-\alpha'\circ\Psi,\Gamma\>\le \frac 12\|\alpha\circ\Phi-\alpha'\circ\Psi\|_\diamond\|\Gamma\|^\diamond=\epsilon\|\Gamma\|^\diamond,
\]
so that (i) holds.

Finally, suppose that  $\Psi\succeq_{\epsilon,\Be}\Phi$. Then by putting  $\alpha=id_\Be$ in (iv), 
we obtain that there is some 
$\alpha'\in \Ce(\Be',\Be)$ such that
$\|\Phi-\alpha'\circ\Psi\|_\diamond\le 2\epsilon$. Let $\De$ be any C*-algebra and let $\Gamma\in \Pe^*(\De,\Ae)$, 
$\phi\in \Ce(\Be,\De)$. Then $\phi'= \phi\circ\alpha'\in \Ce(\Be',\De)$ satisfies
\begin{align*}
\<\phi\circ\Phi-\phi'\circ\Psi,\Gamma\>&=\<\phi\circ(\Phi-\alpha'\circ\Psi),\Gamma\>=\<\Phi-\alpha'\circ\Psi,\Gamma\circ\phi\>\\
&\le \frac12\|\Phi-\alpha'\circ\Psi\|_\diamond\|\Gamma\circ\phi\|^\diamond\le \epsilon\|\Gamma\|^\diamond,
\end{align*}
where  we used the fact that $\Pe^*\circ\Pe\subseteq \Pe^*$ and Corollary \ref{coro:app_norm} (iii) for the first inequality, and 
 Lemma \ref{lemma:monot} for the second. This shows  that $\Psi\succeq_\epsilon \Phi$, the opposite implication is clear.

\end{proof}

In the most important case, $\Phi$ and $\Psi$ are completely positive and by definition,  CP is special among all admissible families of cones.
We will establish some relations between the corresponding deficiencies if $\Phi,\Psi\in \Ce_\Pe\cap CP$.

\begin{coro} Let $\Phi,\Psi\in \Ce_{\Pe}\cap CP$. Then  $\Psi\succeq_{0, CP}\Phi$ implies 
$\Psi\succeq_{0,\Pe}\Phi$. If $CP\subseteq \Pe$, then $\Psi\succeq_{\epsilon, CP}\Phi$ implies 
$\Psi\succeq_{\epsilon,\Pe}\Phi$ for all $\epsilon\ge 0$.
\end{coro}

\begin{proof} Suppose $\Psi\succeq_{0,CP}\Phi$. Let $\Gamma\in \Pe^*(\De,\Ae)$ and $\phi\in\Ce_{\Pe}(\Be,\De)$, then 
$\Gamma\circ\phi\in \Pe^*\cap CP$ and by the assupmtion, there is some $\alpha\in \Ce_{CP}(\Be',\Be)$ such that
\[
\<\phi\circ\Phi,\Gamma\>=\<\Phi,\Gamma\circ\phi\>\le \<\alpha\circ\Psi,\Gamma\circ\phi\>=\<\phi\circ\alpha\circ\Psi,\Gamma\>.
\]
The first assertion now follows from the fact that $\phi\circ\alpha\in \Ce_{\Pe}(\Be',\De)$. The second statement is proved similarly, using the fact that $\|\Gamma\circ\phi\|^\diamond_{CP}\le \|\Gamma\|^\diamond_{CP}\le\|\Gamma\|^\diamond_\Pe$ by Lemma \ref{lemma:monot}
 and  (\ref{eq:diamondu_sup}).

\end{proof}

\subsubsection{Purely quantum and classical deficiency}

There are two important types of decision problems. If $\De=B(D)$ for a Hilbert space $D$, 
deficiency will be called \emph{purely quantum}. 
It is clear that this depends only on $d_D$, so we will write $\Psi\succeq_{\epsilon,d_D}\Phi$ instead of $\Psi\succeq_{\epsilon, B(D)}\Phi$. 

On the other hand, if $\De=\diag_k$, deficiency will be called \emph{classical}. Any $\Gamma\in \Pe^*(\diag_k,\Ae)$ is a cq-map and any 
$\phi\in \Ce(\Be,\diag_k)$ is a qc-channel. It is easy to see that classical deficiency does not depend on the choice of the cone $\Pe$.

It is clear that purely quantum deficiency implies $\De$-deficiency for $\De\subset B(D)$, in particular classical deficiency of the given dimension. In general, there is little hope for the opposite implication. We next  show that purely quantum $k$-deficiency  follows from
$\diag_{k^2}$-deficiency  if $\Pe=CP$ and the channels are tensored 
 with a suitable fixed  channel. These results were inspired by the works of Shmaya \cite{shmaya2005comparison} an Buscemi \cite{buscemi2012comparison}.
In the rest of this paragraph, we assume that $\Pe=CP$. 

Let $d_D=k$. Let $\mathcal G_D:=\{U_j^D,\ j=1,\dots,k^2\}$ 
be a group of unitary operators in  $B(D)$ such that 
\[
\sum_{j} \frac 1{k}(U_j^D)^*aU^D_j=(\Tr a) I_D,\qquad a\in B(D)
\] 
Let $\theta^D_j$ denote the map in $CP(D,D)$, given by  
\[
\theta^D_j: a\mapsto \frac 1{ k}(U_j^D)^*aU_j^D,\quad j=1,\dots, k^2
\]

\begin{lemma}\label{lemma:pqcl} 
Let $\Gamma\in CP(D,H)$ and let $\Phi^{cq}_\Gamma:=\Phi^{cq}_{\{G_j, j=1,\dots,k^2\}}$, where 
$G_j:=(id_H\otimes \theta^D_j)(C(\Gamma))$. Then
\[
\bigl\{\<\phi,\Gamma\>,\ \phi\in \Ce(H,D)\bigr\}=\bigl\{ \<\Phi^{qc}_M, \Phi^{cq}_\Gamma\>,\ M\in \povm(H\otimes D, k^2\bigr\}.
\]
In particular, $\|\Gamma\|^\diamond=\|\Phi^{cq}_\Gamma\|^\diamond$.

\end{lemma}

\begin{proof} Let $M\in \povm(H\otimes D,k^2)$. By (\ref{eq:s_cq})
\begin{align*}
\<\Phi^{qc}_M,\Phi^{cq}_\Gamma\>= \sum_j\Tr M_jG_j=\Tr\sum_j(id_H\otimes (\theta^D_j)^*)(M_j)C(\Gamma)=\Tr Y_MC(\Gamma),
\end{align*}
where $Y_M=\sum_j(id_H\otimes (\theta^D_j)^*)(M_j)$.
Since $Y_M\in B(H\otimes D)^+$ and
\[
\Tr_DY_M=\sum_j\frac 1{n}\Tr_D M_j=\frac 1{n}\Tr_DI_{H\otimes D}=I_H,
\]
there is some unital map $\psi\in CP(D,H)$ such that $Y_M=C(\psi)$. Then  $\phi:=\psi^*\in \Ce(H,D)$ and by 
Remark \ref{rem:choi},
\[
\Tr Y_MC(\Gamma)=\Tr C(\phi^*)C(\Gamma)=\<\phi,\Gamma\>.
\]
Conversely, let $\phi\in \Ce(H,D)$ and let
\[
M_j:=(id_H\otimes(\theta^D_j)^*)(C(\phi^*)), \qquad j=1,\dots,k^2.
\]
Then $M_j\ge 0$ and $\sum_j M_j=\Tr_DC(\phi^*)\otimes I_D=I_{H\otimes D}$, so that $M_j$ is a POVM and we have
\[
\<\phi,\Gamma\>=\Tr C(\phi^*)C(\Gamma)=\sum_jM_jG_j=\<\Phi^{qc}_M,\Phi^{cq}_\Gamma\>.
\]
The last statement follows by Corollary \ref{coro:app_norm} (ii) (Appendix) and Lemma \ref{lemma:base}.

\end{proof}

We now prove some relations between purely quantum and classical deficiency. Note that  full equivalence is proved only in the case $\Be\subseteq B(D)$.

\begin{thm}\label{thm:clpq} Let $\Phi\in \Ce(\Ae,\Be)$, $\Psi\in \Ce(\Ae,\Be')$ and let $D$ be a Hilbert space with $d_D=k$. 
\begin{enumerate}
\item[(i)]  Let  $\xi\in \Ce(H',D)$ be a  surjective channel. Then
$\Psi\otimes \xi\succeq_{0,\diag_{k^2}}\Phi\otimes\xi$ implies  $\Psi\succeq_{0,k} \Phi$.
\item[(ii)] For any $\epsilon\ge 0$, $\Psi\otimes id_D\succeq_{\epsilon,\diag_{k^2}} \Phi\otimes id_D$ implies  $\Psi\succeq_{\epsilon, k}\Phi$.
\item[(iii)] $\Psi\succeq_{\epsilon}\Phi \implies \Psi\otimes \xi\succeq_{\epsilon} \Phi\otimes \xi$
for any channel $\xi$ and any $\epsilon\ge 0$.
\end{enumerate}

\end{thm}

\begin{proof}  Assume that $\Psi\otimes \xi\succeq_{0,\diag_{k^2}}\Phi\otimes\xi$. By Theorem \ref{thm:post} (iii), 
it follows that for any $F=\{F_1,\dots,F_{k^2}\}\subset \Ae_h\otimes B_h(H')$ and $M\in \povm(\Be\otimes B(D),k^2)$
there is some $N\in \povm(\Be'\otimes B(D), k^2)$ such that
\begin{equation}\label{eq:pqcl}
\sum_i \Tr M_i(\Phi\otimes \xi)(F_i) \le \sum_i N_i(\Psi\otimes\xi)(F_i).
\end{equation}
Let $\Gamma\in CP(B(D),\Ae)$ and let   
\[
G_j=(id_\Ae\otimes \theta_j^D)(C(\Gamma))\in (\Ae\otimes B(D))^+,\qquad j=1,\dots,k^2.
\]
Suppose that $\xi$ is surjective, then also $id_\Ae\otimes \xi$ is surjective so that  there is some  
$F_j\in (\Ae_h\otimes B_h(H'))$  such that $G_j=(id_\Ae\otimes \xi)(F_j)$ and  
\begin{equation}\label{eq:gammaeta}
\Phi^{cq}_\Gamma=(id_\Ae\otimes \xi)\circ\Phi^{cq}_F.
\end{equation}
For any C*-algebra $\Fe$  and any channel $\Omega\in \Ce(\Ae,\Fe)$, we have 
\begin{equation}\label{eq:gammaeta_omega}
(\Omega\otimes \xi)\circ\Phi^{cq}_F= (\Omega\otimes id_D)\circ \Phi^{cq}_\Gamma=
\Phi^{cq}_{\Omega\circ\Gamma}.
\end{equation}
Since the RHS is clearly in CP, we obtain by Corollary \ref{coro:app_norm} (ii) that 
\[
\|(\Omega\otimes \xi)\circ\Phi^{cq}_F\|^\diamond =\sup_E \sum_i E_i(\Omega\otimes \xi)(F_i),
\]
where the supremum is taken over all $E\in \povm (\Fe\otimes B(D), k^2)$, so that  (\ref{eq:pqcl}) implies that
$\|(\Phi\otimes \xi)\circ\Phi^{cq}_F\|^\diamond\le \|(\Psi\otimes \xi)\circ\Phi^{cq}_F\|^\diamond$.
By Lemma \ref{lemma:pqcl} and (\ref{eq:gammaeta_omega}), we obtain 
\begin{align*}
\|\Phi\circ\Gamma\|^\diamond=\|\Phi^{cq}_{\Phi\circ\Gamma}\|^\diamond =\|(\Phi\otimes\xi)\circ\Phi^{cq}_F\|^\diamond\le 
\|(\Psi\otimes\xi)\circ\Phi^{cq}_F\|^\diamond
=\|\Psi\circ\Gamma\|^\diamond,
\end{align*}
this proves (i).


Similarly for (ii), suppose that $\Psi\otimes id_D\succeq_{\epsilon,\diag_{k^2}}\Phi\otimes id_D$ and let $\Gamma\in CP(B(D),\Ae)$, then
\begin{align*}
\|\Phi\circ\Gamma\|^\diamond&=\|\Phi^{cq}_{\Phi\circ\Gamma}\|^\diamond=\|(\Phi\otimes id_D)\circ \Phi^{cq}_\Gamma\|^\diamond\\
&\le 
\|(\Psi\otimes id_D)\circ\Phi^{cq}_\Gamma\|^\diamond +\epsilon\|\Phi^{cq}_\Gamma\|^\diamond=\|\Psi\circ\Gamma\|^\diamond+\epsilon\|\Gamma\|^\diamond.
\end{align*}

For (iii), let $\Psi\succeq_{\epsilon}\Phi$, then by Theorem \ref{thm:post} (iv), there is some channel
$\alpha\in \Ce(\Be',\Be)$ such that $\|\Phi-\alpha\circ\Psi\|_\diamond\le 2\epsilon$.  
Let $\xi\in \Ce(H',D)$ and let  $\De$ be any C*-algebra. Let $\phi\in\Ce(\Be\otimes B(D), \De)$ and put $\psi=\phi\circ(\alpha\otimes id_{B(D)})$, then 
$\psi\in \Ce(\Be'\otimes B(D),\De)$ and
we have
\begin{align*}
\|\phi\circ(\Phi\otimes \xi)-\psi\circ(\Psi\otimes \xi)\|_\diamond=\|\phi\circ((\Phi\otimes\xi)-
(\alpha\circ\Psi\otimes\xi))\|_\diamond\\
\le \|(\Phi\otimes\xi)-(\alpha\circ\Psi\otimes\xi)\|_\diamond\le\|\Phi-\alpha\circ\Psi\|_\diamond\le 2\epsilon,
\end{align*}
where the first equality follows from Lemma \ref{lemma:monot} and the second from Theorem  \ref{thm:diamond}. The proof now follows from Theorem \ref{thm:post} (iv).

\end{proof}

\subsubsection{Statistical experiments and pointwise deficiency}\label{sec:statexp}

We  now return to statistical experiments, still assuming that $\Pe=CP$. Let  $\Ee\subset\states(\Be)$ and $\Fe\subset\states(\Be')$. It is easy to see by the remarks at the end of Section \ref{sec:comparison} that
  $\Phi^{cq}_\Ee\succeq_{\epsilon, \De}\Phi^{cq}_\Fe$ is equivalent to $\Ee\succeq_{\epsilon,\De}\Fe$.
Note that by the proof of Theorem \ref{thm:post}, the maximal payoffs are given by 
\[
\max_{\alpha\in \Ce(\Be,\De)} P_\Ee(\alpha, G)=\|\Phi_{G,\Ee}\|_{CP}^\diamond.
\]
This was  was observed already in \cite{jencova2013base}. Theorem \ref{thm:post} (iii) also implies that using  a  broader definition of a decision space where  payoff operators  are only required to be Hermitian leads to an equivalent  notion of deficiency.

Let now $\Ee_0\subset \states(D)$ be an experiment that spans $B(D)$. Then $|\Ee_0|=d_D^2$ and 
$\Phi^{cq}_{\Ee_0}$ is a surjective channel $B(D\otimes D)\to B(D)$. 
By the results of the previous paragraph, we obtain the following slight generalization of \cite{buscemi2012comparison}:

\begin{coro} $\Ee\otimes \Ee_0\succeq_{0,\diag_{d_D^2}}\Fe\otimes \Ee_0$ implies 
$\Ee\succeq_{0,d_D}\Fe$. If $\Be\subset B(D)$, the opposite implication also holds.

\end{coro}

Classical decision problems for quantum statistical experiments have a clear operational meaning, while the generalization to quantum ones seems to be just a mathematical extension.    We next show that for any quantum decision problem there is a classical one  having the same average payoffs.

\begin{prop} Let $(\Ee,B(D),G)$ be a quantum decision problem, with $\Ee\subset \states(H)$ and $|\Ee|=n$. 
Then there is a classical decision problem $(\Ee\otimes \Ee_0,\tilde D,f)$, 
with  $|\tilde D|=d_D^2$ and $f=(f_{ij,l})$, $f_{ij,l}\in\mathbb R$, $i=1,\dots,n$, $j,l=1,\dots,d_D^2$ such that 
\[
\bigl\{P_\Ee(\alpha,G), \ \alpha\in \Ce(H,D)\bigr\}=\bigl\{\sum_iP_{\Ee\otimes\Ee_0}(i, M,f),\ M\in \povm(H\otimes D,d_D^2)\bigr\}
\]

\end{prop}

\begin{proof}  Let $\Ee_0=\{\sigma^0_1,\dots,\sigma^0_{d_D^2}\}$ and let $G=(G_1,\dots,G_n)$. Let 
 $f_{ij,l}\in \mathbb R $ be such that 
\[
\theta^D_j(G_i^t)=\sum_l f_{ij,l}\sigma^0_l,\quad i=1,\dots,n,\ j,l=1,\dots,d_D^2.
\]
Then we have (\ref{eq:gammaeta}) and (\ref{eq:gammaeta_omega}), with $F_j=\sum_{i,l}f_{ij,l}
|i\>\<i|\otimes |l\>\<l|$, $\xi=\Phi^{cq}_{\Ee_0}$, $\Gamma=\Phi^{qc}_G$ and $\Omega=\Phi^{cq}_\Ee$. The proof now follows by Lemma \ref{lemma:pqcl}.

\end{proof}

More generally, let $\Pe$ be any admissible   family of cones. It might be useful to write Theorem \ref{thm:post} for this case, emphasizing
the cone $\Pe$ to show how the conditions differ according to the choice of $\Pe$.

\begin{coro} The following are equivalent.
\begin{enumerate}
\item[(i)] $\Phi_\Ee^{cq}\succeq_{\epsilon,\De,\Pe}\Phi^{cq}_\Fe$
\item[(ii)] For any  $G=\{G_1,\dots,G_n\}\subset \De^+$,
\[
\|\Phi_{G,\Fe}\|^\diamond_\Pe\le \|\Phi_{G,\Ee}\|^\diamond_\Pe+\epsilon\sum_i\|G_i\|
\]
\item[(iii)] For every $\alpha'\in \Ce_\Pe(\Be', \De)$ there is some $\alpha\in \Ce_\Pe(\Be,\De)$ such that 
\[
\sup_i\|\alpha(\sigma_i)-\alpha'(\rho_i)\|_1\le 2\epsilon.
\]

\end{enumerate}

\end{coro}

\begin{rem}\label{rem:weaker} Note that the  condition in (ii) does not have to be checked for all collections of payoff operators. 
For example, we may always assume that the operators  are not invertible. Indeed, suppose $\lambda_j$ is the 
smallest eigenvalue of $G_j$ and let $G'=\{G'_j:=G_j-\lambda_jI_\De, j=1,\dots,n\}$. Then $G_j'$ are not invertible, 
$ \|G_j'\|\le \|G_j\|$ and for any experiment $\Ee$ and $\phi\in \Ce(\Be,\De)$,
\[
\<\phi\circ\Phi^{cq}_\Ee,\Phi^{qc}_{G}\>=\<\phi\circ\Phi^{cq}_\Ee,\Phi^{qc}_{G'}\>+\sum_j\lambda_j.
\]
In particular, for $\De\subseteq B(\mathbb C^2)$, it is enough to assume that $G_i=|g_i\>\<g_i|$ are rank one operators.
This shows that one can restrict to a smaller set of decision spaces. An interesting question is whether it is enough to consider only commuting sets of  payoff operators, that is, whether one can restrict to classical decision spaces. It can be shown \footnote{Private communication with K. Matsumoto} that that this is 
in general not possible even in the case that $\Pe=Pos$. However, by \cite{aluh1980problem}, this is true for $\epsilon=0$ and $\Ae=\Be=\Be'=\De=B(\mathbb C^2)$.

\end{rem}

Let now $\Phi\in \Ce(\Ae,\Be)$ and let $\Gamma\in EB(\De,\Ae)$. Since $EB\subset \Pe^*$, $\Gamma$ is a payoff map.  
We may write $\Gamma=\Phi_{G,\Ee}=\Phi^{cq}_\Ee\circ\Phi^{qc}_G$ for some experiment $\Ee\subset \states(\Ae)$ and some sequence of operators   $G\subset \De^+$.  
For any decision rule $\alpha\in \Ce(\Be,\De)$ we have for the corresponding payoff
\[
P_\Phi(\alpha,\Gamma)=\<\alpha\circ\Phi\circ\Phi^{cq}_\Ee,\Phi^{qc}_G\>=\<\alpha\circ\Phi_{\Phi(\Ee)}^{cq},\Phi^{qc}_G\>=
P_{\Phi(\Ee)}(\alpha,G),
\]
as illustrated on  the following diagram: 
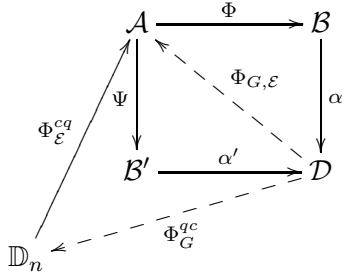
\[
\xymatrix@C=2pc@R=1.5pc{ & \Ae \ar[rr]^\Phi \ar[dd]_\Psi & &\Be \ar[dd]^\alpha\\
 & & & \\
 & \Be'\ar[rr]^{\alpha'} & &\De \ar@{-->}[uull]_{\Phi_{G,\Ee}} \ar@{-->}[dlll]^{\Phi^{qc}_G}\\
\diag_n \ar[uuur]^{\Phi^{cq}_\Ee} & & }
\]
note that the payoff maps are represented by dashed arrows.
Conversely, if $\Ee\subset \states(\Ae)$ is any experiment and  $G\subset \De^+$ a sequence of payoff operators, then $\Phi_{G,\Ee}\in EB(\De,\Ae)$ is a payoff map with $P_{\Phi(\Ee)}(\alpha,G)=P_\Phi(\alpha,\Phi_{G,\Ee})$ for any decision rule $\alpha$. Since $\|\Phi_{G,\Ee}\|^\diamond \le \|\Phi^{qc}_G\|^\diamond$ by Lemma \ref{lemma:monot}, this proves the following:

\begin{prop}\label{prop:pw} $\Psi\succeq_{\epsilon,\De}\Phi$  implies that $\Psi(\Ee)\succeq_{\epsilon,\De}\Phi(\Ee)$ for any experiment
 $\Ee\subset\states(\Ae)$. 

\end{prop}

We call the 'only if' part of this proposition \emph{pointwise $(\epsilon,\De)$-post-processing deficiency}. 
For $\epsilon=0$, the following result is also easy to see.

\begin{prop} The following are equivalent.
\begin{enumerate}
\item[(i)] $\Psi\succeq_{0,\De}\Phi$.
\item[(ii)] $\|\Phi\circ\Gamma\|^\diamond\le \|\Psi\circ\Gamma\|^\diamond$ for all $\Gamma\in EB(\De,\Ae)$.
\item[(iii)] $\Psi(\Ee)\succeq_{0,\De}\Phi(\Ee)$ for any experiment
 $\Ee\subset\states(\Ae)$.
\item[(iv)] $\Psi(\Ee)\succeq_{0,\De}\Phi(\Ee)$ for some experiment $\Ee\subset \states(\Ae)$ that spans $\Ae$. 

\end{enumerate}

\end{prop}

Combining this with the previous paragraph, we obtain the result of Chefles \cite{chefles2009quantum}:

\begin{coro}
Let $\Pe=CP$, $\Phi\in \Ce(H,K)$, $\Psi\in \Ce(H,K')$. The following are equivalent.
\begin{enumerate}
\item[(i)] $\Phi=\alpha\circ\Psi$ for some channel $\alpha\in \Ce(K',K)$.
\item[(ii)] For any Hilbert space $D$, any experiment $\Ee\subset \states(H\otimes D)$ and any $N\in \mathbb N$, $(\Psi\otimes id_D)(\Ee)\succeq_{0,\diag_N} (\Phi\otimes id_D)(\Ee)$.
\item[(iii)] $(\Psi\otimes id_K)(\Ee)\succeq_{0,\diag_{d_K^2}} (\Phi\otimes id_K)(\Ee)$ for some experiment $\Ee$ that spans $B(H\otimes K)$. 
\end{enumerate}

\end{coro}

\subsubsection{Post-processing deficiency for  POVMs}

 Let $M\in \povm(\Ae,m)$. Then $\Phi^{qc}_M$  describes the corresponding  \emph{measurement} on $\Ae$ with $m$ outcomes, in the sense that it maps each state $\rho\in \states(\Ae)$ 
to the vector  of probabilities of the outcomes. A post-processing of $\Phi^{qc}_M$ is a composition with a cq-channel $\Phi^{cq}_\Ee$, resulting in the EB-channel $\Phi_{M,\Ee}$.  Moreover, by a  classical  post-processing of $\Phi^{qc}_M$ 
we obtain a qc-channel $\Phi^{qc}_{\Lambda(M)}$, where $\Lambda$ is a 
 $k\times m$ \emph{stochastic matrix} $\Lambda=(\lambda_{ij})$ (which is a matrix with nonnegative entries, such that $\sum_i\lambda_{ij}=1$ for all $j$) and  $\Lambda(M)_i=\sum_j\lambda_{ij}M_j$.  

Let $N\in \povm(\Ae,n)$. We define post-processing deficiency of $N$ with respect to $M$ as the corresponding deficiency of the qc-channels.  Since $id_{\diag_m}\in \Pe(\diag_m,\diag_m)$ for any $\Pe$, we obtain by the last 
part of  Theorem \ref{thm:post} that for any choice of 
the positive cone, $N\succeq_\epsilon M$ if and only if $N\succeq_{\epsilon, \diag_m} M$, equivalently, there is an $m\times n$ stochastic matrix $\Lambda$ such that $\|\Phi^{qc}_M-\Phi^{qc}_{\Lambda(N)}\|_\diamond\le 2\epsilon$. In particular, for $\epsilon=0$, this means that $N$ is  \emph{post-processing cleaner} than $M$, \cite{bkdpw2005clean}.

\subsection{Pre-processings}

 Suppose that $\Pe$ is  an admissible family of cones and let $\Phi\in \Ce(\Ae,\Be)$. Clearly, there is another possibility to define a
  ''randomization'' of $\Phi$, namely by pre-composition with  some  $\beta\in \Ce(\Ae',\Ae)$. The resulting 
  channel $\Phi\circ\beta\in \Ce(\Ae',\Be)$ will be   called a \emph{pre-processing} of $\Phi$.

A \emph{pre-processing decision space} is a pair $(\De,\Gamma)$, where $\Gamma\in \Pe^*(\Be,\De)$ and the triple $(\Phi,\De,\Gamma)$ with  $\Phi\in \Ce(\Ae,\Be)$ will be called a \emph{pre-processing decision problem}. 
Decision rules are given by elements of $\Ce(\De,\Ae)$. The definition of pre-processing deficiency is formulated in the same way 
as for post-processings, we  use the notation $\Psi\succeq^{\epsilon,\De}\Phi$ meaning that $\Psi$ is \emph{$(\epsilon,\De)$-pre-processing deficient}  with respect to $\Phi$ and $\Psi\succeq^\epsilon \Phi$ if \emph{$\Psi$ is $\epsilon$-pre-processing deficient} with respect to $\Phi$.
The corresponding diagram is  as follows:
\[
\xymatrix@C=4pc@R=3pc{\Ae\ar[r]^\Phi & \Be \ar@{-->}[dl]_\Gamma\\
 \De\ar[u]^\beta\ar[r]_{\beta'} & \Ae'\ar[u]_\Psi}
\]
We omit the proof of the next theorem as it is practically the same as the proof of Theorem \ref{thm:post}.
\begin{thm}\label{thm:pre}\textbf{(Pre-processing)} Let 
$\Phi\in \Ce(\Ae,\Be)$ and $\Psi\in \Ce(\Ae',\Be)$, $\epsilon \ge 0$ and let $\De$  be any  C*-algebra.
  The following are equivalent.
\begin{enumerate}
\item[(i)] $\Psi\succeq^{\epsilon,\De}\Phi$.
\item [(ii)] For any  $\Gamma\in \Pe^*(\Be,\De)$, 
\[
\|\Gamma\circ\Phi\|^\diamond\le \|\Gamma\circ\Psi\|^\diamond+ \epsilon \|\Gamma\|^\diamond
\]
\item [(iii)] For any $\Gamma\in \Le(\Be,\De)$ and all $\beta\in \Ce(\De,\Ae)$, there is some $\beta'\in \Ce(\De,\Ae')$ such that
\[
\<\Phi\circ\beta,\Gamma\>\le \<\Psi\circ\beta',\Gamma\>+2\epsilon \|\Gamma\|^\diamond
\]
\item [(iv)] For any $\beta\in \Ce(\De,\Ae)$ there is some $\beta'\in \Ce(\De,\Ae')$ such that 
\[
\|\Phi\circ\beta-\Psi\circ\beta'\|_\diamond\le 2\epsilon .
\]
\end{enumerate}
Moreover, $\Psi\succeq^{\epsilon} \Phi$ if and only if $\Psi\succeq^{\epsilon,\Ae}\Phi$.
\end{thm}

\subsubsection{Classical pre-processing deficiency}

We next show that unlike post-processings, the \emph{classical pre-processing deficiency} is independent from the size of the set of decisions.
So let $\Phi\in \Ce(\Ae,\Be)$, $\Psi\in \Ce(\Ae',\Be)$ and let $\epsilon \ge 0$. We first observe that $\Psi\succeq^{\epsilon, \diag_1}\Phi$ means that  for every $\sigma\in \states(\Ae)$ there is some 
$\rho\in \states (\Ae')$ such that $\|\Phi(\sigma)-\Psi(\rho)\|_1\le 2\epsilon$, or in other words,
\begin{equation}\label{eq:pre_sub}
\sup_{\sigma\in\Se(\Ae)}\inf_{\rho\in \Se(\Ae')}\|\Phi(\sigma)-\Psi(\rho)\|_1\le 2\epsilon.
\end{equation}
This  condition  means that the range of $\Phi$ is not far from being a subset of the range of $\Psi$, in particular, we obtain
 $\Phi(\states(\Ae))\subseteq \Psi(\states(\Ae'))$ for $\epsilon=0$. We will therefore use the simpler notation $\Psi\supseteq_\epsilon \Phi$ if 
 (\ref{eq:pre_sub}) holds.  Note also that $\Psi\supseteq_\epsilon\Phi$ and simultaneously $\Psi\subseteq_\epsilon\Phi$ if and only if 
\[
\mathrm{dist}_1(\Phi(\states(\Ae)),\Psi(\states(\Ae')))\le 2\epsilon,
\]
where $\mathrm{dist}_1$ is the distance of the two ranges with respect to the trace norm. 

\begin{coro}\label{coro:cl_pre} Let $\Phi\in \Ce(\Ae,\Be)$, $\Psi\in \Ce(\Ae',\Be)$, $\epsilon\ge 0$, $n\in \mathbb N$. The following are equivalent.
\begin{enumerate}
\item[(i)] $\Psi\succeq^{\epsilon,\diag_n}\Phi$.
\item[(ii)] For all $G\in \Be^+$, $\|\Phi^*(G)\|\le \|\Psi^*(G)\|+\epsilon \|G\|$.
\item[(iii)] $\Psi\supseteq_\epsilon \Phi$.

\end{enumerate}

\end{coro}

\begin{proof} Let $(\diag_n,\Gamma)$ be a decision space.
Any decision rule for $(\Phi,\diag_n,\Gamma)$ is   a cq-channel $\Phi^{cq}_\Fe$,  
for some experiment $\Fe=\{\sigma_1,\dots,\sigma_n\}\subset \states(\Ae)$, similarly for $\Psi$. By definition,  
$\Psi\succeq^{\epsilon,\diag_n}\Phi$ if and only if for any such $\Fe$ there is an experiment $\Ee=\{\rho_1,\dots,\rho_n\}\subset\states(\Ae')$ such that 
\[
2\epsilon\ge \|\Phi\circ\Phi^{cq}_{\Fe}-\Psi\circ\Phi^{cq}_\Ee\|_\diamond=\|\Phi^{cq}_{\Phi(\Fe)}- \Phi^{cq}_{\Psi(\Ee)}\|_\diamond=\sup_i \|\Phi(\sigma_i)-\Psi(\rho_i)\|_1.
\]
But this is equivalent with the same condition for $n=1$, which is $\Psi\supseteq_\epsilon \Phi$. By  Theorem \ref{thm:pre}, this is equivalent to (ii) . 
\end{proof}

\subsubsection{Classical and purely quantum pre-procesing deficiency}

We will show the relation between classical and \emph{purely quantum pre-processing deficiency} for $\Pe=CP$, this will be assumed throughout this paragraph. As before, we write $\Psi\succeq^{\epsilon,d_D}\Phi$ instead of $\Psi\succeq^{\epsilon,B(D)}\Phi$.
%
Recall  the notation $|x_\sigma\>=\sum |i_H\>\otimes \sigma^{1/2}|i_H\>\in H\otimes H$, for 
$\sigma\in \states(H)$.

\begin{lemma}\label{lemma:bi_state} Let $\sigma\in \states(H\otimes D)$. Then there is some pure state $\sigma_0\in \states(D\otimes D)$ and some $\beta\in \Ce_{CP}(D,H)$ such that $\sigma=(\beta\otimes id_D)(\sigma_0)$.

\end{lemma}

\begin{proof} Put  $\sigma_D:=\ptr_H \sigma$ and let $p=\mathrm{supp} (\sigma_D)$. Then $\mathrm{supp}(\sigma)\le (I\otimes p)$ and $C':=(I\otimes \sigma_D^{-1/2})\sigma(I\otimes \sigma_D^{-1/2})$ is a positive element in $B(H\otimes D)$, where the inverse is taken only on the support of $\sigma_D$. 
Put $C=C'+(I_H\otimes(I_D- p))$, then $C\ge 0$ and
 $\ptr_H C=p+I_D-p=I_D$, hence there is some $\beta\in \Ce_{CP}(D,H)$ such that $C=C(\beta)$. Moreover,
\[
\sigma=(I_H\otimes\sigma_D^{1/2})C(I_H\otimes \sigma_D^{1/2})=(\beta\otimes id_D)(\sigma_0),
\]
where $\sigma_0=(I_D\otimes \sigma_D^{1/2})X_D(I_D\otimes \sigma_D^{1/2})=|x_{\sigma_D}\>\<x_{\sigma_D}|$ is a pure state in $\states(D\otimes D)$.

\end{proof}

Now we can prove the main results of this paragraph.

\begin{thm}\label{thm:pre_cp_0} Suppose $\Pe=CP$ and let $k\in \mathbb N$. The following are equivalent.
\begin{enumerate}
\item[(i)] $\Psi\succeq^{0,k} \Phi$.
\item[(ii)] $\Psi\otimes\xi\supseteq_0\Phi\otimes\xi$ for any channel $\xi\in \Ce(D,H')$, where $D$ and $H'$ are 
Hilbert spaces, $d_D=k$.

\item[(iii)] $\Psi\otimes\xi\supseteq_0\Phi\otimes\xi$ for some injective $\xi\in \Ce(D,H')$, where $D$ and $H'$ are Hilbert spaces, $d_D=k$.
\end{enumerate}

\end{thm}

\begin{proof} Suppose (i) and let $\xi\in \Ce(D,H')$ be any channel. Let $\sigma\in \states(\Ae\otimes B(D))$, then by 
Lemma \ref{lemma:bi_state} there is some $\beta\in \Ce(B(D),\Ae)$ and a (pure) state $\sigma_0\in \states(D\otimes D)$ such that $\sigma=(\beta\otimes id)(\sigma_0)$. 
 By Theorem \ref{thm:pre} (iv) there is some $\beta'\in \Ce(B(D),\Ae')$ such that $\Phi\circ\beta=\Psi\circ\beta'$.
Hence
\[
(\Phi\otimes \xi)(\sigma)=(\Phi\circ\beta\otimes \xi)(\sigma_0)=(\Psi\circ\beta'\otimes\xi)(\sigma_0)=(\Psi\otimes\xi)(\rho),
\]
where $\rho=(\beta'\otimes id_D)(\sigma_0)\in \states(\Ae'\otimes B(D))$, this shows (ii). The implication (ii) $\implies$ (iii) is trivial.

Suppose (iii) and let $\Gamma\in CP(\Be,B(D))$, $\beta\in \Ce(B(D),\Ae)$. By Remark \ref{rem:choi},
\[
\<\Phi\circ\beta,\Gamma\>=\Tr C(\Phi\circ\beta)C(\Gamma^*)=\Tr (\Phi\otimes id_D)(C(\beta))C(\Gamma^*).
\]
Since $\xi$ is injective, $\xi^*$ and therefore also $id_\Be\otimes \xi^*$ is surjective, hence there is some 
$G\in (\Be\otimes B(H'))_h$ such that $(id_\Be\otimes\xi^*)(G)=C(\Gamma^*)$. Put $\sigma=\tfrac1dC(\beta)$, we obtain
\[
\<\Phi\circ\beta,\Gamma\>=d\Tr (\Phi\otimes \xi)(\sigma)G.
\]
Let $\rho\in \states(\Ae'\otimes B(D))$ be such that $(\Phi\otimes \xi)(\sigma)=(\Psi\otimes \xi)(\rho)$. 
Note that $\xi(\ptr_{\Ae'}(\rho))=\xi(\ptr_{\Ae}\sigma)=\tfrac1d\xi(I_D)$. Since $\xi$ is injective, this implies that
 $d\rho=C(\beta')$ for some $\beta'\in \Ce(B(D),\Ae')$ and
\[
\<\Phi\circ\beta,\Gamma\>=d\Tr(\Psi\otimes\xi)(\rho)G=\Tr (\Psi\otimes id)(C(\beta'))C(\Gamma^*)=\<\Psi\circ\beta',\Gamma\>,
\]
this implies (i).

\end{proof}

For $0<\epsilon <1$ we get a less satisfactory result.

\begin{thm}\label{thm:pre_cp_eps} Suppose $\Pe=CP$ and let $k\in \mathbb N$, $\epsilon \ge 0$. Consider the following statements.
\begin{enumerate}
\item[(i)] $\Psi\succeq^{\epsilon,k} \Phi$.
\item[(ii)] $\Psi\otimes\xi\supseteq_\epsilon\Phi\otimes\xi$ for any channel $\xi\in \Ce(D,H')$, where $D$ and $H'$ are 
Hilbert spaces, $d_D=k$.

\item[(iii)] $\Psi\otimes id_D\supseteq_\epsilon\Phi\otimes id_D$ for $d_D=k$.

\end{enumerate}
Then (i) $\implies$ (ii) $\implies$ (iii) $\implies \Psi\succeq^{\epsilon',k} \Phi$, where $\epsilon'=\epsilon+\frac12\sqrt{\epsilon}$.

\end{thm}

\begin{proof} The proof (i) $\implies$ (ii) is very similar to the proof in the case $\epsilon=0$. 
Let $\xi\in \Ce(D,H')$ be any channel and  let $\sigma\in \states(\Ae\otimes B(D))$, $\sigma=(\beta\otimes id)(\sigma_0)$.
Let $\beta'\in \Ce(B(D),\Ae')$ be such that $\|\Phi\circ\beta-\Psi\circ\beta'\|_\diamond\le 2\epsilon$ and 
put $\rho:=(\beta'\otimes id_D)(\sigma_0)$. Then $\rho\in \states(\Ae'\otimes B(D))$ and
\begin{align*}
\|(\Phi\otimes \xi)(\sigma)&-(\Psi\otimes \xi)(\rho)\|_1\\
&=\|(\Phi\circ\beta\otimes \xi)(\sigma_0)-(\Psi\circ\beta'\otimes \xi)(\sigma_0)\|_1\\
&\le \|\Phi\circ\beta-\Psi\circ\beta'\|_\diamond\le 2\epsilon.
\end{align*}
The implication (ii) $\implies$ (iii) is trivial.

Suppose (iii), let $\Gamma\in CP(\Be,B(D))$, $\beta\in \Ce(B(D),\Ae)$. By Lemma \ref{lemma:cp_udiam} there is some
 $\gamma\in \states(D)$ such that $G:=(I\otimes\gamma^{-1/2})C(\Gamma^*)(I\otimes\gamma^{-1/2})$ satisfies 
$\|\Gamma\|^\diamond=\|G\|$.  By Remark \ref{rem:choi} we have
\[
\<\Phi\circ\beta,\Gamma\>=\Tr (\Phi\otimes id)(C(\beta))C(\Gamma^*)=\Tr (\Phi\otimes id)(\sigma)G,
\]
where $\sigma=(I\otimes\gamma^{1/2})C(\beta)(I\otimes\gamma^{1/2})=(\beta\otimes id)(|x_\gamma\>\<x_\gamma|)\in\states(\Ae\otimes B(D))$.
Let $\rho\in \states(\Ae'\otimes B(D))$ be such that $\|(\Phi\otimes id)(\sigma)-(\Psi\otimes id)(\rho)\|_1\le 2\epsilon$.
By Lemma \ref{lemma:bi_state}, there is some $\beta'\in \Ce(B(D),\Ae')$ and a pure state 
$\rho_0=|x_{\rho_D}\>\<x_{\rho_D}|\in \states(D\otimes D)$  such that $\rho=(\beta'\otimes id)(\rho_0)$. We have
\begin{align*}
\<\Phi\circ\beta,\Gamma\>-\<\Psi\circ\beta',\Gamma\>&=\Tr[(\Phi\otimes id)(\sigma)-(\Psi\otimes id)(\rho)]G+\\&+
\Tr(\Psi\circ\beta'\otimes id)(|x_{\rho_D}\>\<x_{\rho_D}|-|x_\gamma\>\<x_\gamma|)G\\&\le \epsilon\|G\|+\frac12\||x_{\rho_D}\>\<x_{\rho_D}|-|x_\gamma\>\<x_\gamma|\|_1\|G\|
\end{align*}
Further,
\begin{align*}
\||x_{\rho_D}\>\<x_{\rho_D}|-|x_\gamma\>\<x_\gamma|\|_1&=\sqrt{1-|\<x_{\rho_D},x_{\gamma}\>|^2}=\sqrt{1-\Tr \rho_D^{1/2}\gamma^{1/2}}\\
&=\sqrt{\frac12\Tr(\rho_D^{1/2}-\gamma^{1/2})^2}\\
&\le \sqrt{\frac12\|\rho_D-\gamma\|_1}\le \sqrt{\epsilon},
\end{align*}
where the first inequality follows by Powers-St\"ormer inequality \cite{post1970free} and the second inequality follows from 
\[
\|\rho_D-\gamma\|_1=\|\ptr_\Be[(\Phi\otimes id)(\sigma)-(\Psi\otimes id)(\rho)]\|_1\le 2\epsilon.
\]

\end{proof}

\subsubsection{POVMs and pointwise pre-processing deficiency}

The POVMs, or the qc-channels, will play a similar role for pre-processings as the cq-channels, or experiments, for post-processings.

Let $M\in \povm(\Ae,n)$ and let  $(\Phi^{qc}_M,\De,\Gamma)$ be a pre-processing decision problem. Then
 $\Gamma=\Phi^{cq}_G$ for some $G=\{G_1,\dots,G_n\}\subset \De^+$. In this case, we will use a simpler notation $(M,\De,G)$ for this decision problem.  If $\beta\in \Ce(\De,\Ae)$ is a decision rule, we will denote the corresponding payoff by $P_M(\beta,G)=P_{\Phi^{qc}_M}(\beta,\Phi^{cq}_G)$ and if $N\in \povm(\Ae')$ we will write 
 $N\succeq^{\epsilon,\De} M$ instead of $\Phi^{qc}_N\succeq^{\epsilon,\De}\Phi^{qc}_M$. 
As  before, let $\Pe$ be an admissible  family of cones. As in the case of experiments, we write Theorem \ref{thm:pre} for POVMs, indicating the cone $\Pe$.

\begin{coro} The following are equivalent.
\begin{enumerate}
\item[(i)] $N\succeq^{\epsilon,\De, \Pe}M$.
\item[(ii)] For all $G=\{G_1,\dots,G_n\}\subset \De^+$. 
\[
\|\Phi_{M,G}\|^{\diamond}_\Pe\le \|\Phi_{N,G}\|^\diamond_\Pe+\epsilon \inf_{\sigma\in \states(\De)}\max_i\|\sigma^{-1/2}G_i\sigma^{-1/2}\|.
\]
\item[(iii)] For all $\beta\in \Ce_\Pe(\De,\Ae)$ there is some $\beta'\in \Ce_\Pe(\De,\Ae')$ such that 
\[
\sup_{\sigma\in \states(\De)}\sum_i\|\sigma^{1/2}(\beta^*(M_i)-(\beta')^*(N_i))\sigma^{1/2}\|_1\le 2\epsilon.
\]
\end{enumerate}

\end{coro}

In the case that $\Pe=CP$, $\epsilon=0$ and $\De=\Ae$, these are equivalent conditions for $N$ to be \emph{cleaner} than $M$, in the sense of  
\cite{bkdpw2005clean}. Note also that Theorem \ref{thm:pre_cp_0} for POVMs yields  \cite[Theorem 7.2.]{bkdpw2005clean}.
In fact, the former can be obtained from the latter, using the relation of pre-processing deficiency with its 
pointwise version for $\epsilon=0$ proved below.

Let $\Phi\in \Ce(\Ae,\Be)$ be a channel and consider the decision problem $(\Phi,\De,\Gamma)$ where
 $\Gamma=\Phi_{F,G}$ is an EB map. By Remark \ref{rem:EB}
 we may suppose that $F$ is a POVM. For any decision rule $\beta$, we have for
 the corresponding payoff 
\begin{align*}
P_\Phi(\beta,\Gamma)=\<\Phi\circ\beta,\Phi_{F,G}\>=\<\Phi^{qc}_{\Phi^*(F)}\circ\beta,\Phi^{cq}_G\>=P_{\Phi^*(F)}(\beta, G).
\end{align*}
Conversely, let $F\in \povm(\Be,n)$, then $\Phi^*(F)\in \povm(\Ae,n)$. Let $(\Phi^*(F),\De,G)$ be a decision problem, then for any decision rule $\beta$, $P_{\Phi^*(F)}(\beta,G)=P_\Phi(\beta, \Phi_{F,G})$.
It follows that, similarly as in the case of post-processings, decision problems with the payoff represented by an EB map are closely related to
\emph{pointwise pre-processing deficiency} of channels, see also the diagram
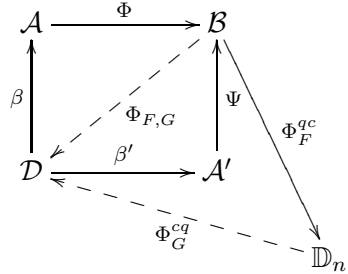
\[
\xymatrix@C=2pc@R=1.5pc{ \Ae \ar[rr]^\Phi & &\Be \ar@{-->}[ddll]^{\Phi_{F,G}} \ar[dddr]^{\Phi^{qc}_F}& \\
 & & & \\
 \De \ar[rr]^{\beta'} \ar[uu]^\beta & &\Ae' \ar[uu]_{\Psi} & \\
 & & & \diag_n \ar@{-->}[ulll]^{\Phi^{cq}_G}}
\]
Similarly as for post-processings, we obtain:

\begin{prop} $\Psi\succeq^{\epsilon,\De,\Pe}\Phi$ implies that $\Psi^*(F)\succeq^{\epsilon,\De,\Pe} \Phi^*(F)$ for any 
 $F\in \povm(\Be)$.

\end{prop}

Let now $E\in \povm(\Be,n)$ be \emph{informationally complete}, that is, $\Tr E_i\rho=\Tr E_i \sigma$ for all $i$ implies that $\rho=\sigma$ for any pair of states $\rho,\sigma\in \states(\Be)$, equivalently,  $\mathrm{span}(E)=\Be$.
Let $\Phi\in \Ce(\Ae,\Be)$ and $\Psi\in \Ce(\Ae',\Be)$, then $\Phi^*(E)\in \povm(\Ae,n)$ and $\Psi^*(E)\in \povm(\Ae',n)$. 
Suppose that
$\Psi^*(E)\succeq^{0,\De}\Phi^*(E)$ and let $\beta\in \Ce(\De,\Ae)$, then there is some $\beta'\in\Ce(\De,\Ae')$ such that
 $\beta^*(\Phi^*(E_i))=(\beta')^*(\Psi^*(E_i))$ for all $i$. Since $E$ is informationally complete, this implies that
$\beta^*\circ\Phi^*=(\beta')^*\circ\Psi^*$, that is, $\Phi\circ\beta=\Psi\circ\beta'$, so that we have proved:

\begin{prop} The following are equivalent.
\begin{enumerate}
\item[(i)] $\Psi\succeq^{0,\De}\Phi$.
\item[(ii)] $\Psi^*(E)\succeq^{0,\De}\Phi^*(E)$ for any $E\in \povm(\Be)$.
\item[(iii)]  $\Psi^*(E)\succeq^{0,\De}\Phi^*(E)$ for some informationally complete $E\in\povm(\Be)$.

\end{enumerate}

\end{prop}
 In particular, for 0-deficiency it is enough to consider decision problems with EB payoff maps.

\subsubsection{Pre-processing deficiency for statistical experiments}

We finish by description of pre-processing deficiency for statistical experiments. 
This time, $\Ee$ and $\Fe$ are families of states in the same 
algebra $\Be$, but the number of elements is different, say, $\Ee=\{\rho_1,\dots,\rho_n\}$ and $\Fe=\{\sigma_1,\dots,\sigma_m\}$. We define pre-processing deficiency of  experiments as pre-processing deficiency of $\Phi^{cq}_\Ee$ and $\Phi^{cq}_{\Fe}$. Since the input algebra is $\Ae=\diag_m$, we obtain by Theorem \ref{thm:pre} and Corollary \ref{coro:cl_pre}
that $\Ee\succeq^\epsilon \Fe$ if and only if $\Phi^{cq}_\Ee\supseteq_\epsilon\Psi^{cq}_\Fe$. Since the range 
of $\Phi^{cq}_\Ee$ is the convex hull $co(\Ee)$, we obtain 

\begin{coro}  $\Ee\succeq^\epsilon \Fe$ if and only if $\sup_i\inf_{\rho\in co(\Ee)}\|\sigma_i-\rho\|_1\le 2\epsilon$.

\end{coro}

\section{Final remarks and questions}

We have shown that the randomization theorem for quantum statistical experiments  fits naturally into the framework of base section norms, more precicely the diamond norm and its dual. This allowed us to extend the theory to more 
general classes of  channels, even in the case that the channels and decision rules  are not necessarily completely 
positive but  positivity is given by an admissible family of cones. This extension is purely mathematical and it is not clear how to interpret the post- and pre- processing  decision problems. However, it works quite nicely for characterizing the situation when a channel is not far from a pre- or post-processing of another channel.
Moreover,  it is clarified that  both classical and quantum randomization theorems are a consequence of 
duality of the two norms.
 We finish with a list of remarks and
questions that were left for future work.

 The norm $\|\cdot\|_{\diamond,\Pe}$ is interesting in its own right. As follows from the results of \cite{jencova2013base}, 
this is a \emph{distinguishability norm} for elements in $\Ce_\Pe(H,K)$, in the sense that the minimum Bayes error probability for symmetric hypothesis testing of $\Phi_0$ against $\Phi_1$ is given by 
$\frac12(1-\frac12\|\Phi_0-\Phi_1\|_{\diamond,\Pe})$. Here \emph{tests} are defined as afine maps 
$\Ce_\Pe(H,K)\to [0,1]$, assigning to $\Phi$ the probability $p(\Phi)$ of rejecting $\Phi_0$ if $\Phi$ is true. It can be infered that such tests are given by triples $(H_0,\rho, M)$, where $\rho\in \states(H\otimes H_0)$ and $0\le_{\tilde \Pe}
M\le_{\tilde \Pe} I$ is an element in $B(K\otimes H)$, see Remark 
\ref{rem:norms}, such that $p(\Phi)=\Tr M(\Phi\otimes id_{H_0})(\rho)$. If $\Pe=CP$, these are called \emph{1-testers} \cite{cdp2008memory} or 
\emph{PPOVMs} \cite{ziman2008PPOVM}. If $\Phi_0$ and $\Phi_1$ are in $CP\cap \Pe$, $\|\cdot\|_{\diamond,\Pe}$ gives the minimum Bayes error probability if the tests are restricted to those satisfying the above inequalities.

 In \cite{aluh1980problem}, it was shown that for pairs of qubit states, 0-deficiency is equivalent to classical 0-deficiency.
This suggests that the set of payoff functionals in condition (ii) of Theorem \ref{thm:post} can be further restricted, see also 
Remark \ref{rem:weaker}.

 For the relation of purely quantum and classical post- and pre-processing deficiency, it seems that a full equivalence should hold in Theorems \ref{thm:clpq} and \ref{thm:pre_cp_eps}.

  Another question is if one can restrict to EB payoffs for $\epsilon$-deficiency also if $\epsilon >0$. 
A closely related  question is if pointwise $\epsilon$-deficiency is equivalent to $\epsilon$-deficiency.
It seems that this should hold at least if the payoffs are restricted to EB maps, the question is if one can always write such a map in the form $\Gamma=\Phi_{G,\Ee}$ with $\Ee\subset \states(\Ae)$ and 
$\|\Gamma\|^\diamond=\|\Phi^{qc}_G\|^\diamond$, resp. $\Gamma=\Phi_{F,G}$ with $F\in \povm(\Be)$ and $\|\Gamma\|^\diamond=\|\Phi^{cq}_G\|^\diamond$. 

 One can also consider more general randomizations than post- and pre-processings. For example, the composition
$\Phi\mapsto \alpha\circ\Phi\circ\beta$ where $\alpha$ and $\beta$ are channels is also a transformation between channels.
If $\Pe=CP$, it was shown in \cite{cdp2008supermaps} that most general physical transformations that map $\Ce(H,K)$ to $\Ce(H',K')$
 are given by \emph{quantum supermaps}, defined as follows: let $H_0$ be an ancilla and let $\alpha\in \Ce(H',H\otimes H_0)$ and $\beta\in \Ce(K\otimes H_0,K')$, then the transformation has the form 
\[
\Phi\mapsto \beta\circ(\Phi\otimes id_{H_0})\circ\alpha
\]
It should be possible to prove similar results also for this kind of randomizations, using the dual pair of norms
$\|\cdot\|_{2\diamond}$ and $\|\cdot\|^{2\diamond}:=\|\cdot\|_{2\diamond}^*$, see \cite{gutoski2012measure, jencova2013base}. Even more generally, one could prove randomization theorems for quantum networks using the norms $\|\cdot\|_{n\diamond}$ and 
$\|\cdot\|^{n\diamond}$, obtaining a framework for questions like how close two networks of length $n$ can get in an appropriate norm when applied to networks of length $m$, etc.

 The methods used here work only in finite dimensions. If one of the spaces is infinite dimensional, different methods have to be developed.

\section{Appendix: Base sections and norms in ordered vector spaces}

In \cite{jencova2013base}, a family of norms was introduced in the ordered vector space $B_h(H)$ endowed with the positive cone $B(H)^+$. This family includes base norms with respect to any base of the positive cone, as well as  order unit norms with respect to any order unit in $B_h(H)$.  We now introduce similar norms for any (finite dimensional) ordered vector space with a proper cone. 
 For some basic definitions and facts we need in the sequel, see Section 1.2 in \cite{jencova2013base}, or 
\cite{rockafellar1970convex}. 

We will consider finite dimensional real ordered vector spaces $(\Ve,Q)$, with a proper cone $Q$.
We denote the  partial order in $\Ve$ by $\le$ and the dual order in $\Ve^*$ by $\le^*$.

A subset $B\subset \Ve$ will be called a \emph{base section} in $(\Ve,Q)$ if $B\cap int(Q)\ne \emptyset$ and $B=\Te\cap S$ for some linear subspace $\Te\subseteq \Ve$ and a base $S$ of $Q$. It is easy to see that then 
$B=\mathrm{span}(B)\cap S$. This gives the following characterization of base sections.

\begin{lemma}\label{lemma:app_base} Let $S$ be a base of $Q$. If $B\subset S$ contains an element of $int(Q)$, then $B$ is a base section
 if and only if  for any $t,s\ge 0$ and $b_1,b_2\in B$, $tb_1-sb_2\in S$ implies that $tb_1-sb_2\in B$.
\end{lemma}

Since $\Te$ contains an interior point of $Q$, $\Te\cap Q$ is a proper cone in $\Te$ and it is clear that $B$ is its base. The  relative interior of $B$ is the set of order units contained in $B$,
$ri(B)=B\cap int(Q)$.

If $B$ is a base section, we define the \emph{dual} of $B$ as 
\[
\tilde B:=\{ b^*\in Q^*, \<b,b^*\>=1\}.
\]
\begin{lemma}\label{lemma:app_dual} $\tilde B$ is a base section in $(\Ve^*,Q^*)$ and $\tilde{\tilde B}=B$.

\end{lemma}

\begin{proof} Let $B=\Te\cap S$, then there is some $\tilde b\in int(Q^*)$ such that $S=\{q\in Q,\ \<q,\tilde b\>=1\}$ and it is clear  that $\tilde b\in \tilde B$. Moreover, let $b\in ri(B)$, then $b$ is an order unit and $\tilde B$ is contained in the corresponding base $S^*$ of $Q^*$. Let now $s,t\ge 0$ and $\tilde b_1,\tilde b_2\in \tilde B$ be such that $t\tilde b_1-s\tilde b_2\in S^*$. Then $t-s=1$ and we have $t\tilde b_1-s\tilde b_2\in\tilde B$. By Lemma \ref{lemma:app_base}, $\tilde B$ is a base section. 

It si clear that $B\subseteq \tilde{\tilde B}$ and  it is easy to see that $\tilde B=(\tilde b+ B^\perp)\cap Q^*$, where $B^\perp$ is the annihilator of $B$. Similarly, $\tilde{\tilde B}=(b+\tilde B^\perp)\cap Q$. Since $\tilde b\in int(Q^*)$, for every $x^*\in B^\perp$ there is some $t>0$ such that $e^*\pm tx^*\in \tilde B$. It follows that $\tilde B^\perp\subseteq (B^\perp)^\perp=\mathrm{span}(B)$, hence 
$\mathrm{span}(\tilde{\tilde B})= \mathrm{span}(B)$. But then $B$ and $\tilde{\tilde B}$ are two bases of the cone $\mathrm{span}(B)\cap Q$, so that $B=\tilde{\tilde B}$.

\end{proof}

\begin{lemma}\label{lemma:app_krein} For any  affine functional $f:B\to \mathbb R^+$ there is some $q^*\in Q^*$ such  that
$f(b)=\<b,q^*\>$ for all $b\in B$.

\end{lemma}

\begin{proof} Since $B$ is a base of $\mathrm{span}(B)\cap Q$, $f$ extends to a positive linear functional on $(\mathrm{span}(B), \mathrm{span}(B)\cap Q)$. As $B\cap int(Q)$ is nonempty, Krein's theorem \cite{naimark1959normed} implies that this functional extends to a positive functional on $(\Ve,Q)$, that is an element of $Q^*$.

\end{proof}

For a base section $B$, we define
\[
\Oe_B:=\{q_1-q_2, \ q_1,q_2\in Q,\ q_1+q_2\in B\}=\{x\in \Ve, \exists b\in B,\ -b\le x\le b\}
\]
\begin{thm}\label{thm:app_norm}
$\Oe_B$ is the unit ball of a norm $\|\cdot\|_B$ in $\Ve$. The dual norm in $\Ve^*$ is $\|\cdot\|_{\tilde B}$.

\end{thm}

\begin{proof} The proof is similar to that of Theorem 1 in \cite{jencova2013base}, the first part is the same.
For the second part, we have to prove that $\Oe_{\tilde B}=\Oe_B^\circ$. 
So let $v^*=q_1^*-q_2^*$, $q_i^*\in Q^*$, $q_1^*+q_2^*\in \tilde B$ and let $-b\le x\le b$ for some $b\in B$, then
\[
\<x,v^*\>=\<x,q_1^*\>-\<x,q_2^*\>\le\<b,q_1^*+q_2^*\>=1,
\]
so that $\Oe_{\tilde B}\subseteq \Oe_B^\circ$. For the opposite inclusion, let $\Ve_2=\Ve\oplus\Ve$, with proper cone  $Q_2=Q\oplus Q$. Suppose
$B=S\cap \Te$ and let $b\in ri(B)$, $\tilde b\in ri(\tilde B)$. Let 
\[
B_2=\{q_1\oplus q_2\in Q_2, \ q_1+q_2\in B\},
\]
then $\tfrac12 (b\oplus b)\in B_2\cap int(Q_2)$ and  $B_2=S_2\cap \Te_2$, where $S_2$ is the base of $Q_2$ corresponding to 
$\tilde b\oplus\tilde b\in int(Q^*)$ and 
$\Te_2=\{x\oplus y, x+y\in \Te\}$. Hence $B_2$ is a base section. It is clear that if $v^*\in \Oe_B^\circ$, then 
$\<w,v^*\oplus -v^*\>\le 1=\<w,\tilde b\oplus\tilde b\>$ for all $w\in B_2$, so that the functional 
$a^*:=(\tilde b-v^*)\oplus(\tilde b+v^*)$ is positive on $B_2$. By Lemma \ref{lemma:app_krein}, there is some element in $Q^*_2$ that coincides with $a^*$ on $B_2$, it follows that there is some $x^*\in B_2^\perp$ such that $a^*+x^*\in Q_2^*$. It is not difficult to show that any element $x^*\in \Te_2^\perp$ has the form $x^*=y^*\oplus y^*$ for some $y^*\in\Te^\perp$. We obtain that $\pm v^*\le^* \tilde b+y^*$, so that sor any $q\in Q$, $ \pm\<q,v^*\>\le\<q,\tilde b+y^*\>$. This implies that we must have 
$\tilde b+y^*\in (\tilde b+\Te^\perp)\cap Q^*=\tilde B$. Hence $v^*\in \Oe_{\tilde B}$.

\end{proof}

We next observe that the norm $\|\cdot\|_B$ is a generalization of base norms and order unit norms. In fact, if $B=S$ is a base of $Q$, then $\|\cdot\|_S$ is the corresponding base norm. On the other hand, if $B=\{b\}$ for some order unit $b$, $\|\cdot\|_B=\|\cdot\|_b$ is the order unit norm.

\begin{coro}\label{coro:app_norm} 
\begin{enumerate}
\item[(i)] If $x\in \Ve$, then
$\|x\|_B=\inf_{B\in ri(B)}\|x\|_b=\sup_{\tilde b\in ri(\tilde B)}\|x\|_{S_{\tilde b}}$,
 where $S_{\tilde b}$ is the base of $Q$ corresponding to $\tilde b$. 
\item[(ii)] If $x\in Q$, then
$\|x\|_B=\sup_{\tilde b\in \tilde B} \<x,\tilde b\>$.
\item[(iii)] If $b,b'\in B$, then $\sup_{q^*\in \Oe_{\tilde B}\cap Q^*}\<b-b',q^*\>=\tfrac12\|b-b'\|_B$.
\end{enumerate}

\end{coro}

\begin{proof} (i) and (ii) are proved exactly as \cite{jencova2013base}. 
Let $b,b'\in B$ and suppose $q^*\in \Oe_{\tilde B}\cap Q^*$. Then $0\le^* q^*\le^* \tilde b$ for
 some  $\tilde b\in \tilde B$. Put $p^*=2(q^*-\tfrac12 \tilde b)$, then $p^*\in \mathcal O_{\tilde B}$ and 
$q^*=\tfrac12(p^*-\tilde b)$. Since $\<b-b', \tilde b\>=0$, we obtain
\[
\<b-b',q^*\>=\frac12 \<b-b',p^*\>\le \frac12\|b-b'\|_B.
\]
Conversely, let $p^*\in \mathcal O_{\tilde B}$, then there is some $\tilde b\in \tilde B$ such that
$q^*=\tfrac12(p^*+\tilde b)\in \mathcal O_{\tilde B}\cap Q^*$, so that 
\[
\frac12 \<b-b',p^*\>=\<b-b',q^*\>\le 
\sup_{\tilde q\in \mathcal O^{\tilde B}\cap Q^*}\<b-b',\tilde q\>.
\]
This proves (iii).

\end{proof}

Let $(\Ve_1,Q_1)$ and $(\Ve_2,Q_2)$ be two ordered vector spaces and  let $B_i\subset Q_i$ be base sections. Let $T:\Ve_1\to \Ve_2$ be a linear map. Then $T$ is \emph{positive}  if 
$T(Q_1)\subseteq Q_2$. It is clear that $T$ is positive if and only if its adjoint $T^*$ is positive and that  $T(B_1)\subseteq B_2$ if and only if $T^*(\tilde B_2)\subseteq \tilde B_1$. Let $\|T\|_{B_1,B_2}$ be the norm of $T$ with respect to the norms $\|\cdot\|_{B_i}$ in $\Ve_i$. 

\begin{prop}\label{prop:app_maps} Let $T:\Ve_1\to \Ve_2$ be a positive linear map. Then
$\|T\|_{B_1,B_2}=\sup_{b\in B_1}\|T(b)\|_{B_2}$.

\end{prop}

\begin{proof} Let  the supremum on the RHS be equal to $s$. Let $\|x\|_{B_1}\le 1$, so that there is some $b\in B_1$ such that $-b\le x\le b$.
Then $-T(b)\le T(x)\le T(b)$ and $\|T(b)\|_{B_2}\le s$. It follows that $\|T(x)\|_{B_2}\le s\le \sup_{\|x\|_{B_1}\le 1}\|T(x)\|_{B_2}=\|T\|_{B_1,B_2}$.

\end{proof}

\section*{Acknowledgement}  
This work was supported by the grants  VEGA 2/0125/13 and by Science and Technology Assistance Agency under the contract no. APVV-0178-11.


\end{document}